\documentclass[sigconf]{acmart}

\usepackage{booktabs} % For formal tables
\usepackage{algpseudocode,algorithm,algorithmicx,subfigure}
\usepackage{epstopdf}
\usepackage{thmtools, thm-restate}

\def \uell{u_{\ell}}
\def \uHat{\hat{U_{\ell}}}
\def \rHat{\hat{R}}
\newcommand{\argmin}{\mbox{${arg\,min}$}} 
\newcommand{\Prob}{\mbox{${\rm Pr}$}} 
\newcommand{\sN}{\mathcal{N}} 
\newcommand{\sS}{\mathcal{S}}

\newtheorem{problem}{Problem}
\newtheorem{claim}{Claim}
\newcommand{\calA}{\mathcal{A}}
\newcommand{\calB}{\mathcal{B}}  
\newcommand{\calC}{\mathcal{C}}  
\newcommand{\calD}{\mathcal{D}}
\newcommand{\calI}{\mathcal{I}} 
\newcommand{\calM}{\mathcal{M}} 
 
\newcommand{\calU}{\mathcal{U}} 
\newcommand{\calW}{\mathcal{W}}
\newcommand{\lambdab}{\lambda^{\mathcal{B}}} 
\newcommand{\lambdabprime}{\lambda^{\mathcal{B}'}} 
\newcommand{\lambdaa}{\lambda^{\mathcal{A}}} 
\newcommand{\cali}{\mathcal{I}}
\newcommand{\gtilde}{\tilde{g}}
\newcommand{\ftilde}{\tilde{f}}

\newcommand{\eat}[1]{} 
\newcommand*\Let[2]{\State #1 $\gets$ #2}
\algrenewcommand\algorithmicrequire{\textbf{Input:}}
\algrenewcommand\algorithmicensure{\textbf{Output:}}
\algdef{SE}[DOWHILE]{Do}{doWhile}{\algorithmicdo}[1]{\algorithmicwhile\ #1}%

\setcopyright{rightsretained}
\begin{document}
\title{Combating the Cold Start User Problem in Model Based Collaborative Filtering}
%\titlenote{Produces the permission block, and
  %copyright information}
%\subtitle{Extended Abstract}
%\subtitlenote{The full version of the author's guide is available as
 % \texttt{acmart.pdf} document}

\author{Sampoorna Biswas}
%\authornote{Dr.~Trovato insisted his name be first.}
\affiliation{%
  \institution{University of British Columbia}
  %\streetaddress{2366 Main Mall}
  %\city{Vancouver} 
  %\state{BC} 
  %\postcode{V6T 1Z4}
}
\email{sambis@cs.ubc.ca}

\author{Laks V.S. Lakshmanan}
%\authornote{The secretary disavows any knowledge of this author's actions.}
\affiliation{%
  \institution{University of British Columbia}
 % \streetaddress{2366 Main Mall}
  %\city{Vancouver} 
  %\state{BC} 
  %\postcode{V6T 1Z4}
}
\email{laks@cs.ubc.ca}

\author{Senjuti Basu Ray}
%\authornote{This author is the one who did all the really hard work.}
\affiliation{%
  \institution{New Jersey Institute of Technology}}
 % \streetaddress{1 Th{\o}rv{\"a}ld Circle}
  %\city{Hekla} 
  %\country{Iceland}}
\email{senjutib@njit.edu}

\renewcommand{\shortauthors}{S. Biswas et al.}

\begin{abstract}
For tackling the well known cold-start user problem in model-based recommender systems, one approach is to recommend a few items to a cold-start user and use the feedback to learn a profile. The learned profile can then be used to make good recommendations to the cold user. In the absence of a good initial profile, the recommendations are like random probes, but if not chosen judiciously, both bad recommendations and too many recommendations may turn off a user. We formalize the cold-start user problem by asking what are the $b$ best items we should recommend to a cold-start user, in order to learn her profile most accurately, where $b$, a given budget, is typically a small number. We formalize the problem as an optimization problem and present multiple non-trivial results, including NP-hardness as well as hardness of approximation. We furthermore show that the objective function, i.e., the least square error of the learned profile w.r.t. the true user profile, is neither submodular nor supermodular, suggesting efficient approximations are unlikely to exist. Finally, we discuss several scalable heuristic approaches for identifying 
the $b$ best items to recommend to the user and experimentally evaluate their performance on 4 real datasets. Our experiments show that our proposed accelerated algorithms significantly outperform the prior art in runnning time, while achieving similar error in the learned user profile as well as in the rating predictions.%, on all the datasets tested. 
\end{abstract}
%\begin{abstract}
%This paper provides a sample of a \LaTeX\ document which conforms,
%somewhat loosely, to the formatting guidelines for
%ACM SIG Proceedings\footnote{This is an abstract footnote}. 
%\end{abstract}

\eat{

}

% We no longer use \terms command

%\keywords{ACM proceedings, \LaTeX, text tagging}

\maketitle

\section{Introduction}

%Recommender systems have emerged as a popular solution to the information overload problem and have been successfully deployed in many application domains such as recommendation of products (books, music, movies, etc.), services (e.g., restaurants), and content (e.g., search queries, hashtags, and other online content). Owing to their wide appeal, various approaches have been developed for generating good recommendations that are likely to appeal to their users. One of the most successful approaches is Collaborative Filtering~\cite{ekstrand2011collaborative}. It exploits the intuition that users may like items preferred by users with tastes or interests that are similar to theirs, or that users may enjoy items similar to those that they have enjoyed in the past, or a combination of these ideas. Collaborative filtering (CF) may be either memory-based~\cite{breese1998empirical}, meaning the raw feedback data from users, in the form of ratings, is used for identifying similar users or items, or it may be model-based. 

In order to generate good recommendations, one of the most popular methods in recommender systems is model-based collaborative filtering (CF)~\cite{ekstrand2011collaborative}, which assumes a generative model. An approach that has been particularly successful is the so-called matrix factorization (MF) approach, which assumes a latent factor model of low dimensionality for users and items, which are learned by factoring the matrix of observed ratings~\cite{koren2009matrix}. %Compared to memory-based CF, matrix factorization has been found to result in higher quality recommendations and is more scalable. 
One reason for the success of latent factor models 
%couched in terms of movie recommendation, 
is that the latent factors can capture discriminating hidden features of items and users 
%such as comedy vs drama, amount of violence, orientation to children, etc, 
even when these features are not explicitly available as part of the data or are difficult to obtain. These extracted features are useful for making superior recommendations. As demonstrated by the Netflix prize competition, one of the most sophisticatd realizations of latent factor models is based on MF techniques~\cite{koren2009matrix}. In the rest of this paper, we consider recommender systems based on MF.%, specifically the probabilistic interpretation of the MF model proposed in \cite{pmf}. 

An important challenge faced by any recommender system is the so-called \emph{cold-start user} and \emph{cold-start item} problem. The former occurs when a new user joins the system and the latter when a new item becomes available or is added to the system's inventory. Since the system has very little information on such users and items, CF techniques perform poorly on cold-start users and items. In order to learn a profile or model of a cold-start user, we need to have the user's feedback on a certain minimum number of items, which involves recommending some items to that user. A key question is \emph{how to select items to recommend to a cold-start user}. Active learning strategies try to answer this question, but most approaches that have been explored in the literature have mainly tended to be ad hoc and heuristic in nature~\cite{Golbandi, Zhou-fMF, karimi2015supervised, rubens2015active, rashid2008learning}. While these works report empirical results based experiments conducted on some datasets, unfortunately, these works do not formulate the item selection problem in a rigorous manner and do not analyze its computational properties. Furthermore, no comprehensive scalability experiments have been reported on their proposed strategies for item selection.

{\bf Our Contributions:} In this paper, we focus on the cold-start user problem.% and formulate the item selection problem for cold-start users. 
We assume a latent factor model based on matrix factorization for our underlying recommender system. Since user attention and patience is limited, we assume that there is a budget $b$ on the number of items for which we can request feedback from a cold-start user. The main question we then study is, \emph{how to select the $b$ best items to recommend to such a user that will allow the system to learn the user's profile as accurately as possible}. The motivation is that if the user profile is learned well, it will pay off in allowing the system to make high quality recommendations to the user in the future. We formulate the item selection problem as a discrete optimization problem, called \emph{optimal interview design} (OID), where the items selected can be regarded as questions selected for interviewing the cold-start user for her feedback on those items (Section~\ref{sec:prob-stmt}). 

Our first challenge is in formalizing the problem, i.e., defining the true user profile against which to measure the error of a learned profile. This is necessary for defining the objective function we need to optimize with our choice of $b$ items. The difficulty is that there is no prior information on a cold-start user. We address this by showing that under reasonable assumptions, which will be made precise in Section~\ref{sec:background}, we can directly express the difference between the learned user profile and the true user profile in terms of the latent factors of the $b$ items chosen. This allows us to reason about the quality of different choices of $b$ items and paves the way for our optimization framework (Section~\ref{sec:sfw}). 

Our second challenge is to analyze the problem theoretically. We establish that OID problem is NP-hard. The proof is fairly non-trivial and involves an intricate reduction from Exact Cover by 3-sets (X3C) (Section~\ref{sec:hardness}). We subsequently show that the optimal interview design problem is NP-hard to approximate to within a factor $\frac{\alpha}{\theta}$, where $\alpha$ and $\theta$ depend on the problem instance (Section~\ref{sec:appx}). Furthermore, we show that the objective function, i.e., least squared error between the true and learned user profile, is neither submodular nor supermodular, suggesting efficient approximation algorithms may be unlikely to exist (Section~\ref{sec:supermodularity}).

Our third challenge is computational. Since OID is both NP-hard, hard to approximate, and the objective function is neither submodular or supermodular, we present several heuristic scalable algorithms for selecting the $b$ best items to minimize the error (Section~\ref{sec:approach}). Our empirical results demonstrate that our algorithms significantly outperform previously studied state-of-the art heuristic solutions in scalability, while achieving similar quality in terms of error (Section~\ref{sec:expt}).% with statistical significance. 

Related work is discussed in Section~\ref{sec:related-work}. The necessary background appears in Section~\ref{sec:background}, while Section~\ref{sec:conclusion} summarizes the paper and discusses future work.

\section{Related Work}\label{sec:related-work}
We classify research related to the problem studied in this paper under the following categories. 

\textbf{Cold Start Problem in CF}. The cold-start problem in CF-based recommender systems has been addressed using different approaches in prior work. A common approach combines CF with user content (metadata) and/or item content information to start off the recommendation process for cold users~\cite{%kim2010collaborative,
lam2008addressing,schein2002methods,lika2014facing,zhang2014addressing}. Other approaches leverage information from an underlying social network to recommend items to cold users \cite{massa2007trust,jamali2009trustwalker}. Some researchers have tried to solve it as an active learning problem \cite{rashid2008learning, rubens2015active}.
%similarity measures tuned to cold-start users \cite{ahn2008new,bobadilla2012collaborative}, and use of feature-based regression \cite{park2009pairwise}. 
In addition, online CF  techniques, that incrementally update the latent vectors as new items or users arrive, have been proposed as a way to incorporate new data without retraining the entire model ~\cite{abernethy2007online, sarwar2002incremental, huang2016real}. None of these works rigorously study the problem of selecting a limited number of items for a cold-start user as an optimization problem. 

One exception is~\cite{anava2015budget}, which studies {\em the cold-start item problem} and formalizes it as an optimization problem of selecting users, to rate a given cold-start item. We borrow motivation from this paper and study the {\em cold-start user problem} by formalizing an optimization function in a probabilistic manner. Unlike them, our recommender model is based on probabilistic MF. Furthermore, they do not study the complexity or approximability of the user selection problem in their framework. They also do not run any scalability tests, and their experiments are quite limited. As part of our technical results, we show that our objective function is not supermodular. By duality between the technical problems of cold-start users and cold-start items, it follows that the objective used in their framework is not supermodular either, thus correcting a misclaim in their paper. A practical observation about the cold-start user problem is that it is easy and natural to motivate a cold-start user by asking her to rate several items in return for better quality recommendations using the learned profile. However, it is less natural and therefore harder to motivate users to help the system learn the profile of an item, so that it can be recommended to \emph{other} users in the future.

\textbf{Interactive Recommendation}. Items may be recommended to a cold-start user in batch mode or interactive mode. In batch mode, the items are selected in one shot and then used for obtaining feedback from the cold-start user. E.g., this is the approach adopted in~\cite{anava2015budget} (for user selection). In interactive mode, feedback obtained on an item can be incorporated in selecting the next item. Interactive recommendations are handled in two ways -- offline or online. We focus on the offline approach which considers all possible outcomes for feedback and prepares an ``interview plan'' in the form of a decision tree~\cite{Golbandi, Zhou-fMF, karimi2015supervised}. %It is well known that constructing the optimal decision tree is NP-complete. Furthermore, the search space of possible decision trees for forming interview plans is exponential. 
While heuristic solutions are proposed in~\cite{Golbandi, Zhou-fMF, karimi2015supervised}, large scale scalability experiments are not reported.  In contrast, multi-armed bandit frameworks that interleave exploration with exploitation have been studied \cite{%warmup, bandits2, 
zhao2013interactive, caron2013mixing, bandits3, zeng2016online} in online setting. However, these  approaches %(i.e., bandit approaches) 
require re-training of the model after each item is recommended.

In sum, to the best of our knowledge, {\sl we are the first to formalize the item selection problem for interviewing a cold-start user as a discrete optimization problem, and analyze its complexity and approximability, besides proposing scalable solutions}.

\eat{ 
 The hardness of recommendation systems has also been studied under various assumptions - non-negative matrix factorization, used for collaborative filtering recommender systems \cite{vavasis2009complexity}, maximizing recommendations by selecting a seed set of users and recommending products to them \cite{goyal2012recmax}, forming groups of users such that recommendations made to them maximize satisfaction \cite{basu2015group}, and finding the $b$ best items to recommend that maximizes the probability of consumer purchase \cite{hammar2013using} have all been shown to be hard problems. But none of them address the cold-start problem.

%We plan to address the complexity of selecting the $k$ best items to recommend to a cold-start user in order to learn her profile as accurately as possible. 
%We also plan to develop an efficient solution for this problem. 

%\begin{figure*}
  %\includegraphics[width=\textwidth,height=4cm]{example-decision-tree.PNG}
  %\caption{An example interview plan of depth 3}
%\end{figure*}
} 

\section{Preliminaries \& Problem Statement}
\label{sec:background}
In this section, we summarize the relevant notions on collaborative filtering (CF) and [present further technical development. 

\subsection{Recommender Systems}
\label{sec:rsbg} 

Most recommender systems (RS) use a matrix $R^{m\times n}$ of ratings given by users to some items, with $r_{ij}$ denoting the rating of item $j$ by user $i$. We assume there are $m$ users and $n$ items, and an arbitrary, but fixed rating scale. The goal of CF based on latent factor models is to factor $R$ into a pair of matrices $U \in \mathbb{R}^{d\times m}$ and $V \in \mathbb{R}^{d \times n}$, consisting of low dimensional latent factor vectors of users and items respectively, such that their product approximates $R$ as closely as possible. The learned factor matrices are used to predict unknown ratings: the predicted rating of item $j$ by user $i$, is $\hat{r}_{ij} = U_i^T V_j$. 
Items with high predicted ratings are recommended to users. We denote the matrix of predicted ratings by $\rHat$. Matrix factorization (MF), a popular approach to CF, tries to find factor matrices such that the RMSE between predicted and observed ratings is minimized: i.e., $\argmin_{U,V} ||R - U^TV||^2_F$, where $||A||_F := \sqrt{\sum_{i=1}^m\sum_{j=1}^n a_{ij}^2}$ denotes the Frobenius norm of matrix \cite{koren2009matrix}. 
%%%%% 
\eat{ 
Most recommender systems use a rating/interaction matrix, $R$, of size $ n \times m$ with $n$ users and $m$ items. Each entry $r_{ij}$ in the matrix represents user $i$'s preference rating on item $i$. We assume the rating scale is given to us and is orthogonal to our problem. 
} 
%%%%%% 
%We focus on learning the profile of a single cold user, represented as one row of unknown ratings in $R$.

\subsection{Matrix Factorization}
\label{sec:pmf}

%MF treats the (latent) features of users and items are deterministic quantities. Probabilistic MF (PMF), on the other hand, 
For our underlying recommender system, we look at the probabilistic interpretation of matrix factorization (MF) models which assumes that user and item features are drawn from distributions. More precisely, 
it expresses the rating matrix $R$ as a product of two random low dimension latent factor matrices with the following zero-mean Gaussian priors \cite{pmf}: 

\begin{align}
\vspace*{-1ex} 
\Prob[U|\Sigma_U] =  \prod^m_{i=1} \sN(U_i|\mathbf{0}, \sigma_{u_i}^2I),  
%
%\hspace*{1ex}   
\Prob[V|\Sigma_V] =  \prod^n_{j=1} \sN(V_j|\mathbf{0}, \sigma_{v_j}^2I), 
\end{align}

\vspace*{-1.5ex} 
\noindent 
where $\sN(x|\mu, \sigma^2)$ is the probability density function of a Gaussian distribution with mean $\mu$ and variance $\sigma^2$. 
It then estimates the observed ratings as $R = \hat{R} + \varepsilon = 
U^T V + \varepsilon$, where $\varepsilon$ is a matrix of noise terms in the model. More precisely, $\varepsilon_{ij} = \sigma_{ij}^2$ represents zero-mean noise in the model.

%%%%%%%%%% 
\eat{ 
. Let the two lower dimension matrices be represented by $U \in \mathbb{R}^{d \times m}$ and $V \in \mathbb{R}^{d \times n}$, where each column $U_i$ and $V_j$ represents a specific user and item respectively. The objective is to use the constructed user and item profiles to predict the unknown ratings. Let $\rHat$ represent the matrix of rating predictions. The predicted rating for user $u_i$ and item $v_j$ can be characterized as an interaction between their $d$-dimensional vectors: $\rHat_{ij} = U_i^T V_j + \varepsilon_{ij}$, where $\varepsilon_{ij} = \sigma_{ij}^2$ represents zero-mean noise in the model. %Here we assume $\mathbb{E}[\varepsilon_{ij}^2] = \sigma^2$, but in Section \ref{sec:non-identical-noise} we also consider a different noise model. 
} 
%%%%%%%  

The conditional distribution over the observed ratings is given by 
\begin{equation}
%p(R|U, V, \sigma^2) = \prod^m_{i=1} \prod^n_{j=1} [\sN(R_{ij}|U_i^TV_j, \sigma^2)]^{I_{ij}}
\Prob[R|U, V, \Sigma] = \prod^m_{i=1} \prod^n_{j=1} [\sN(R_{ij}|U_i^TV_j, \sigma_{ij}^2)]^{\delta_{ij}}
\end{equation}

\vspace*{-1.5ex} 
\noindent 
where %$\sN(x|\mu, \sigma^2)$ is the probability density function of a Gaussian distribution with mean $\mu$ and variance $\sigma^2$, 
$\Sigma$ is a $d \times d$ covariance matrix, and $\delta_{ij}$ is an indicator function with value $1$ if user $u_i$ rated item $v_j$, and $0$ otherwise. 

\eat{
 Zero-mean  Gaussian priors are placed on $U, V$ as follows,
\begin{equation}
%p(U|\sigma_U^2) =  \prod^m_{i=1} \sN(U_i|0, \sigma_U^2I),  \hspace{1em}   p(V|\sigma_V^2) =  \prod^n_{j=1} \sN(V_j|0, \sigma_V^2I)
\Prob[U|\Sigma_U] =  \prod^m_{i=1} \sN(U_i|\mathbf{0}, \sigma_{u_i}^2I),  \hspace{1em}   \Prob[V|\Sigma_V] =  \prod^n_{j=1} \sN(V_j|\mathbf{0}, \sigma_{v_j}^2I)
\end{equation}

Taking the log over the posterior gives us (see \cite{pmf}): 
\begin{align}
%\ln p(U, V|R, \sigma^2, \sigma_V^2, \sigma_U^2)  = -\frac{1}{2\sigma^2} \sum_{i=1}^m \sum_{j=1}^n I_{ij} (R_{ij} - U_i^TV_j)^2 \nonumber \\
 %-\frac{1}{2\sigma_U^2}\sum_{i=1}^m U_i^T U_i -\frac{1}{2\sigma_V^2}\sum_{j=1}^n V_j^T V_j \nonumber \\
 %- 1/2 \Bigg( \Bigg( \sum_{i=1}^m \sum_{j=1}^n I_{ij} \Bigg) \ln \sigma^2  + m d \ln \sigma_U^2  + n d \ln \sigma_V^2 \Bigg) + C
\ln \Prob[U, V|R, \Sigma, \Sigma_V, \Sigma_U]  = -\frac{1}{2} \sum_{i=1}^m \sum_{j=1}^n \frac{\delta_{ij}}{\sigma_{ij}^2} (R_{ij} - U_i^TV_j)^2 \nonumber \\
 -\frac{1}{2}\sum_{i=1}^m \frac{U_i^T U_i}{\sigma_{u_i}^2} -\frac{1}{2}\sum_{j=1}^n \frac{V_j^T V_j}{\sigma_{v_j}^2} \nonumber \\
  - \frac{1}{2} \Bigg( \sum_{i=1}^m \sum_{j=1}^n \delta_{ij} \ln \sigma_{ij}^2  + d \sum_{i=1}^m \ln \sigma_{u_i}^2  + d \sum_{j=1}^n \ln \sigma_{v_j}^2 \Bigg) + C
\end{align}
}
%Maximizing the log posterior while keeping the hyperparameters fixed involves optimizing for both $U$ and $V$ simultaneously, which makes it a non-convex optimization problem, with an objective function similar to matrix factorization. 
Algorithms like gradient descent or alternating least squares can be used to optimize the resulting log posterior, which is a non-convex optimization problem.

\subsection{Problem Statement}\label{pd}
\label{sec:prob-stmt} 

Consider a MF model $(U, V)$ trained on an observed ratings matrix $R$, by minimizing a loss function such as squared error between $R$ and the predicted ratings $\hat{R} = U^T V$ (with some regularization). Let $\uell$ be a cold-start user whose profile needs to be learned by recommending a small number of items to $\uell$. Each item $v_j$ recommended to $\uell$ can be viewed as a probe or ``interview question'' to gauge $\uell$'s interest profile. Since there is a natural limit on how many probe items we can push to a user before saturation or apathy sets in, we assume a budget $b$ on the \# probe items. We denote the true profile of $\uell$ by $U_{\ell}$ and the learned profile (using her feedback  on the $b$ items) as $\uHat$. Our objective is to select $b$ items that minimizes the expected error in the learned profile $\uHat$ compared to the true profile $U_{\ell}$. We next formally state the problem studied in this paper. 

\noindent 
\begin{problem}[Optimal Interview Design] 
\label{prob-oid} 
Given user latent vectors $U$, item latent vectors $V$, cold start user $u_{\ell}$, and a budget $b$, find the $b$ best items to recommend to $u_{\ell}$ such that $E[||\uHat - U_{\ell}||^2_F]$ is minimized. 
\end{problem}

\begin {table}[H]
    \caption {Notations Table} \label{tab:notations} 
\begin{center}
%	\label{tab:datasets}
    \begin{tabular}{| l | p{6cm} |}
    \hline
    \textbf{Notation} & \textbf{Interpretation} \\ \hline
    $R^{m\times n}$ & Rating matrix\\ \hline
    $U, V$ & User and item latent factor matrices\\ \hline
    $U_i, V_j$ & Latent vector for user $u_i$ and item $v_j$\\ \hline
    $\hat{R}$ & Matrix of predicted ratings \\ \hline
    $\uell$ & Cold start user \\ \hline
    $\uHat$ & Estimated latent vector of cold user\\ \hline
    $B$ & Recommended items to $\uell$  \\ \hline
    $C$ & Diagonal covariance matrix with $\sigma_1, ..., \sigma_n$ on diagonal positions \\
    \hline
    \end{tabular}
    %\vspace{-0.05in}
\end{center}
\end {table}

%\vspace{-0.05in}

\section{Solution Framework} 
\label{sec:sfw} 
\label{sec:extended-pmf} 
A first significant challenge in solving Problem~\ref{prob-oid} is that in order to measure how good our current estimate the user profile is, we need to know the actual profile of the cold user, on which we have no information! In this section, we devise an approach for measuring the error in the estimated user profile, which intelligently circumvents this problem (see Lemma~\ref{mod}).

Note that using the MF framework described in Section \ref{sec:pmf}, we obtain low dimensional latent factor matrices $U, V$. In the absence of any further information, we assume that the latent vector of the cold user ``truly'' describes her profile.\footnote{There may be a high variance associated with $U_{\ell}$.} Notice that the budget $b$ on the number of allowed interview/probe items is typically a small number. Following prior work~\cite{anava2015budget, rendle2008online, sarwar2002incremental}, we assume that the responses of the cold user $\uell$ to this small number of items does not significantly change the latent factor matrix $V$ associated with items. Under this assumption, we can perform local updates to $U_{\ell}$ as the ratings from $\uell$ on the $b$ probe items are available. 
\eat{ 
Moreover, as we design the interview plan for the cold user, we only update the cold user profile $U_{\ell}$, while keeping the item matrix $V$ unchanged, as the interview plan  contains a rather small set of questions (budgeted by $b$) that are unlikely to make any significant change in $V$. Not only that, this assumption allows us to re-train the PMF model more efficiently after every question. Both of these assumptions are indeed realistic in recommender system and widely adopted in prior work~\cite{anava2015budget}. 
} 
A second challenge is that we consider a \emph{batch setting} for our problem. This means that we should select the $b$ items without obtaining explicit feedback from the cold user. We overcome this challenge by estimating the feedback rating the user $\uell$ would provide according to the current model. Specifically, we estimate cold user $\uell$'s rating on an item $v_j$ as $R_{\ell j} = \hat{R}_{\ell j} + \varepsilon_{\ell j} = V_j^TU_{\ell} + \varepsilon_{\ell j}$, where $\varepsilon_{\ell j}$ is a noise term associated with the user-item pair $(\uell, v_j)$.

Let $R_{\ell}$ denote the vector containing the ratings of the cold user $\uell$ on the $b$ items presented to her, and let $V_B$ be the $d \times b$ latent factor matrix corresponding to these $b$ items. We assume that the noise in estimating the ratings $\rHat$ depends on the item under consideration, i.e., $\mathbb{E}[\varepsilon_{ij}^2] = \sigma_{v_j}^2$, for all users $u_i$. This gives us the following posterior distribution,
$$
\Prob[U_{\ell}|R_{\ell}, V_B, C_B^2] \propto \sN(R_{\ell}|V_B^TU_{\ell}, C_B^2) \sN(U_{\ell}|\mathbf{0}, \sigma_{u_{\ell}}^2I)
$$ 
where $C_B$ is a $b \times b$ diagonal matrix with $\sigma_1, \sigma_2, ..., \sigma_b$ at positions corresponding to the items in $B$. Using Bayes rule for Gaussians, we obtain $\Prob[U_{\ell}|R_{\ell}, V_B, C_B^2] \propto \sN (U_{\ell}|\uHat, \Sigma_B)$, where $\uHat = \Sigma_B V_B C_B^{-2} R_{\ell}$ and $\Sigma_B =  (\sigma^{-2}_{u_{\ell}}I + V_B C_B^{-2} V_B^T)^{-1}$. Setting  $\gamma = \sigma^{-2}_{u_{\ell}}$, the estimate $\uHat$ of the cold user's true latent factor vector $U_{\ell}$ can be obtained using a ridge estimate. More precisely, 
\begin{align}
\uHat = (\gamma I + V_B C_B^{-2} V_B^T)^{-1} V_B C_B^{-2} R_{\ell}  \label{eq:uhat-estimate-non-identical}
\end{align}
Here, $\gamma$ is mainly used to ensure that the expression is invertible.

Under this assumption, we next show that solving Problem~\ref{prob-oid} reduces to minimizing $tr((V_B C_B^{-2} V_B^T)^{-1})$, where $tr(M)$ denotes the trace of a square matrix $M$ i.e., the sum of its diagonal elements. More precisely, we have: 

\begin{lemma}\label{mod} 
Given user latent vectors $U$, item latent vectors $V$, cold start user $\uell$, and budget $b$, a set of $b$ items $B$ minimizes $E[||\uHat - U_{\ell}||^2_F]$ iff it minimizes $tr((V_B C_B^{-1} V_B^T)^{-1})$, where $V_B$ is the submatrix of $V$ corresponding to the $b$ selected items. 
\end{lemma} 

\begin{proof}
Our goal is to select $b$ items such that using her feedback on those items, we can find the estimate of the latent vector $\uHat$ of the cold user $u_{\ell}$, that is as close as possible to the true latent vector vector $U_{\ell}$. 

%\subsubsection{Identically Distributed Noise}
Equation \ref{eq:uhat-estimate-non-identical} gives us an estimate for $\uHat$. For simplicity, we will assume that $\gamma = 0$, and that $V_B C_B^{-2} V_B^T$ is invertible.

$R_{\ell}$ can be expressed as $V_B^TU_{\ell} + \varepsilon_B$, where $\varepsilon_B$ is a vector of the $b$ zero-mean noise terms corresponding to the $b$ items. Replacing this in Equation \ref{eq:uhat-estimate-non-identical}, we get
\begin{align}
\uHat = U_{\ell} + (V_B C_B^{-2} V_B^T)^{-1} V_B C_B^{-2} \varepsilon_B \nonumber \\
\Rightarrow \uHat - U_{\ell} = (V_B C_B^{-2} V_B^T)^{-1} V_B C_B^{-2} \varepsilon_B \label{eq:uhat-ui}
\end{align}

From Equation \ref{eq:uhat-ui}, it is clear that the choice of the $b$ interview items determines how well we are able to estimate $\uHat$.
The expected error in the estimated user profile is 
\begin{equation}
\begin{split}
E[||\uHat-U_{\ell}||_F^2] = E[tr((\uHat-U_{\ell})(\uHat-U_{\ell})^T)] \label{eq:error}
\end{split}
\end{equation}

Replacing Equation \ref{eq:uhat-ui} in Equation \ref{eq:error} and simplifying, we get 
%\begin{multline}
 
\begin{align} 
E[tr((\uHat - U_{\ell})(\uHat - U_{\ell})^T)] = \nonumber \\
E[tr((V_B C_B^{-2} V_B^T)^{-1} V_B C_B^{-2} \varepsilon_B \varepsilon_B^T (C_B^{-2})^T V_B^T (V_B C_B^{-2} V_B^T)^{-1})] \nonumber \\
								 = tr((V_B C_B^{-2} V_B^T)^{-1}) \label{eq:obj-func}
%\end{multline}
\end{align}

The second equality above follows from from replacing $E[\varepsilon_B \varepsilon_B^T] = C_B^2$ and simplifying the algebra. The lemma follows. 
\end{proof}

In view of the lemma above, we can instantiate Problem~\ref{prob-oid} and restate it as follows. 
\noindent 
\begin{problem}[{\bf Optimal Interview Design (OID)}] 
\label{probnew-oid} 
Given user latent vectors $U$, item latent vectors $V$, cold start user $u_{\ell}$, and a budget $b$, find the $b$ best items to recommend to $u_{\ell}$ such that $E[||\uHat - U_{\ell}||^2_F] = tr((V_B C_B^{-2} V_B^T)^{-1})$ is minimized. 
\end{problem} 

For a square matrix $\calM$, we define $f(\calM) := tr((\calM \calM^T)^{-1})$. Note that for our objective function, setting $\calM = V_B C_B$ we get $f(V_B C_B) = tr((V_B C_B^{-2} V_B^T)^{-1})$.
Since the lemma shows that Problem~\ref{prob-oid} is essentially equivalent to Problem~\ref{probnew-oid}, we focus on the latter problem in the rest of the paper.

 \section{Technical Results} 
In this section, we study the hardness and approximation of the OID problem we proposed. 
\subsection{Hardness} 
\label{sec:hardness}

Our first main result in this section is: 

\begin{theorem} 
\label{thm-oid}  
The optimal interview design (OID) problem (Problem ~\ref{probnew-oid}) is NP-hard.  
\end{theorem} 

The proof of this theorem is fairly non-trivial. 
We establish this result by proving a number of results along the way. For our proof, we consider the special case where the items variances are identical, i.e., $\sigma_{v_1}^2 = \sigma_{v_2}^2 = ... = \sigma_{v_n}^2 = \sigma^2$ and $\lambda = \frac{\sigma^2}{\sigma^2_{u_{\ell}}}$. Then $C_B = \sigma I$, and plugging it in to Equation~\ref{eq:obj-func} yields $E[||\uHat-U_{\ell}||_F^2] = \sigma^2 \cdot tr((V_B V_B^T)^{-1})$. We prove hardness for this restricted case. The hardness of the general case follows. 
%Concretely, Problem~\ref{prob-oid} can be reformulated as follows: given a set of latent item vectors $V$ corresponding to items $\cali$, find a subset of items $B\subset \cali$, $|B|=b$, such that $tr((V_BV_B^T)^{-1})$ is minimized, where $V_B$ is the submatrix of $V$ associated with the chosen items $B$. 

The proof is by reduction from the well-known NP-complete problem Exact Cover by 3-Sets (X3C)~\cite{garey1979}.  
 
\noindent  
{\bf Reduction}: Given a collection $\sS$ of 3-element subsets of a set $X$, where $|X| = 3q$, X3C asks to find a subset $\sS^*$ of $\sS$ such that each element of $X$ is in exactly one set of $\sS^*$. Let $(X,\sS)$  be an instance of X3C, with $X = \{x_1, ..., x_{3q}\}$ and $\sS=\{S_1, ..., S_n\}$. Create an instance of OID as follows. Let the set of items be $\cali = \{a_1, ..., a_n, d_1, ..., d_k\}$, where $k=3q$, item $a_j$  corresponds to set $S_j$,  $j\in[n]$, and $d_j$ are dummy items,  $j\in[k]$.  
Convert each set $S_j$ in $\sS$ into a binary vector ${\bf u_j}$ of length $k$, such that ${ \bf a_j}[i] = 1$ whenever $x_i\in S_j$ and ${ \bf a_j}[i] = 0$ otherwise. Since the size of each subset is exactly 3, we will have exactly three 1's in each vector. These vectors correspond to the item latent vectors of the $n$ items $a_1, a_2, ..., a_n$.  We call them \textit{set vectors} to distinguish them from the vectors corresponding to the dummy items, defined next: for a dummy item $d_j$, the corresponding vector ${ \bf d_j}$ is such that ${ \bf d_j}[j] = \eta$ and ${ \bf d_j}[i]=0$, $i\neq j$. Let $\calW$ be the set of all vectors constructed. We will set the value of $\eta$ later. Thus, $\calW$ is the transformed instance obtained from $(X, \sS)$. Assuming an arbitrary but fixed ordering on the items in $\cali$, we can treat $\calW$ as a $k\times (n+k)$ matrix, without ambiguity. Let $\calA = \{\bf{a_1, ..., a_n}\}$ and $\calD = \{\bf{d_1, ..., d_k}\}$  resp., denote the sets of set vectors and dummy vectors constructed above. We set the budget to $b := q+k$ and the item variances $\sigma_{v_1}^2 = ... = \sigma_{v_n}^2 = 1$. For a set of items $B \subset \cali$, with $|B|=b$, we let $\calB$ denote the $k\times (q+k)$ submatrix of $\calW$ associated with the items in $B$. Formally, our problem is to find $b$ items $B\subset \cali$ that minimize $tr((\calB\calB^T)^{-1})$.   
 
For a matrix $\calM$, recall that $f(\calM) = tr((\calM \calM^T)^{-1})$. Define
\begin{align} 
\theta & := \frac{q}{3 +\eta^2} + \frac{k-q}{\eta^2}. \label{eq-theta} 
\end{align} 
We will show the following claim.  
 
\begin{claim}\label{claim1}  
Let $B \subset \cali$, such that $|B| = k + q$. Then $f(\calB) = \theta$ if  $(\calB \setminus \calD)$ encodes an exact 3-cover of $X$ and $f(\calB) > \theta$, otherwise.  
\end{claim}

Notice that Theorem~\ref{thm-oid} follows from Claim~\ref{claim1}: if there is a polynomial time algorithm for solving OID, then we can run it on the reduced instance of OID above and find the $b$ items $B$ that minimize $f(\calB)$. Then by checking if $f(\calB) = \theta$, we can verify if the given instance of X3C is a YES or a NO instance.  
 
In what follows, for simplicity, we will abuse notation and use $\calA, \calB, \calW$ both to denote sets of vectors and the matrices formed by them, relative to the fixed ordering of items in $\cali$ assumed above. We will freely switch between set and matrix notations.  

%%%%%%%%%%%%%%%%%%%%%%%% 
We first establish a number of results which will help us prove the above  claim. 
Recall the transformed instance $\calW$ of OID obtained from the given X3C instance. 
The next claim characterizes the trace of $\calB\calB^T$ for matrices $\calB \subset \calW$ that include all $k$ dummy vectors of $\calW$. 

\begin{claim} 
\label{claim2} 
\label{claim:equal_trace} 
Consider any $\calB\subset \calW$ such that $|\calB| = k + q$ and $\calB$ includes all the $k$ dummy vectors. Then $tr(\calB \calB^T) = k + k \cdot \eta^2$.  
\end{claim} 

\begin{proof} 
Let $\calB' = \calB -\calD$. We have  
$tr(\calB \calB^T)  = tr(\calB' \calB'^T + \calD \calD^T)$  
$= tr(\calB' \calB'^T) + tr(\eta^2I) = \sum_{i=1}^k \sum_{j=1}^q b_{ij}^2 + k\eta^2$.  
As $\calB'$ is a binary matrix, $\sum_{i=1}^k \sum_{j=1}^q b_{ij}^2 = \sum_{i=1}^k ||b_{*i}||_0 $, where $b_{*i}$ is the $i$th row, and $|| \cdot ||_0$ is the $l_0-$norm. This is nothing but the total number of 1's in $\calB'$, which is $3q = k$. 
Thus, $tr(\calB \calB^T) = k +  k\eta^2$. 
\end{proof}
%%%%%%%%%%%%%%%%%%%%%%%%

The next claim shows that among such subsets $\calB\subset\calW$, the ones that include all dummy vectors have the least $f(.)$-value, i.e., have the minimum value of $tr((\calB\calB^T)^{-1})$. Recall that $\calD = \{{\bf d_1, ..., d_k}\}$ is the set of dummy vectors constructed from the given instance of X3C. 

\begin{claim} \label{claim3}  
For any subset $\calA \subset \calW$, with $|\calA| = k+q$, such that  
$\calD \not\subset \calA$, there exists $\calA'$, with $|\calA'|=k+q$ and $\calD\subset \calA'$, such that $f(\calA') < f(\calA)$.  
\end{claim} 
 
\begin{proof} 
By Claim \ref{claim:equal_trace},  $tr(\calA'\calA'^T) = k + k\eta^2$. By assumption, $\calA$ has at least 1 fewer dummy vectors than $\calA'$ and correspondingly more set vectors than $\calA'$. Since each set vector has exactly 3 ones, we have $tr(\calA\calA^T) \leq k + k\eta^2 + 3 - \eta^2$ for $\eta^2 > 3$. Let us consider the way the trace is distributed among the eigenvalues. The distribution giving the least $f(.)$ is the uniform distribution. For $\calA\calA^T$, this is $\lambda_1 = \lambda_2 = ... = \lambda_k = tr(\calA \calA^T)/k$. The distribution yielding the maximum $f(.)$ is the one that is most skewed. For $\calA'\calA'^T$, this happens when there are two distinct eigenvalues, namely $\eta^2$ with multiplicity $(k-1)$ and $k+\eta^2$ with multiplicity $1$. This is because, the smallest possible eigenvalue is $\eta^2$ and the trace must be accounted for.\footnote{Such extreme skew will not arise in reality since this corresponds to all $q$ set vectors of $\calA$ being identical (!), but this serves to prove our result.}

%%%%%%%%%%%%%%%%%%% 
We next show that the largest possible value of $f(\calA')$ is strictly smaller than the smallest possible value of $f(\calA)$, from which the claim will follow. 
 
Under the skewed distribution of eigenvalues of $\calA'\calA'^T$ assumed above, $f(\calA') \le \frac{k-1}{\eta^2} + \frac{1}{k + \eta^2}$. Similarly, for the uniform distribution for the eigenvalues of $\calA\calA^T$ assumed  above, $f(\calA) \geq k \times k/tr(\calA \calA^T) \geq \frac{k^2}{k + k\eta^2 + 3 - \eta^2}$.

%%%%%%%%%%%%%%%%%%%%%%%%% 
Set $\eta$ to be any value $\ge \sqrt{(k+3)}$. Then we have 

\begin{align} 
f(\calA') & \le \frac{(k-1)}{(k+3)} + \frac{1}{(2k+3)} \nonumber \\ 
& = \frac{2k(k+1)}{(k+3)(2k+3)} \nonumber \\ 
f(\calA) & \ge \frac{k^2}{(k+k(k+3)+3 -k -3)} \nonumber \\ 
& = \frac{k}{(k+3)}. \nonumber \\ 
\end{align}
Now, $2(k+1) < (2k+3)$. Multiplying both sides by $k(k+3)$ and rearranging, we get the desired inequality $f(\calA') \le \frac{2k(k+1)}{(k+3)(2k+3)} < \frac{k}{(k+3)} \le f(\calA)$, showing the claim. We can obtain a tighter bound on $\eta$ by solving $\frac{k-1}{\eta^2} + \frac{1}{k + \eta^2} \leq \frac{k^2}{k + k\eta^2 + 3 - \eta^2}$, which gives us $\eta^2 \geq \frac{1}{2}[\sqrt{5k^2 + 4} - k + 4]$.
\end{proof}
 
In view of this, in order to find $\calB \subset\calW$ with $|\calB|=k+q$ that minimizes $f(\calB)$, we can restrict attention to those sets of vectors $\calB$ which include all the $k$ dummy vectors. 

Consider $\calB\subset\calW$, with $|\calB|=k+q$ that includes all $k$ dummy vectors. We will show in the next two claims that the trace $tr(\calB\calB^T)=k+k\eta^2$ will be evenly split among its eigenvalues iff $\calB-\calD$ encodes an exact 3-cover of $X$. We will finally show that it is the even split that leads to minimum $f(\calB)$.  
  
\begin{claim} 
\label{claim4} 
\label{claim:unequal_eigenval} 
Consider a set $\calB$, with $|\calB| = k + q$, such that $\calB$ includes all the $k$ dummy vectors. Suppose the rank $q$ matrix $\calB' = \calB - \calD$ 
does not correspond to an exact 3-cover of $X$. Then $\calB' \calB'^T$ has $q$ non-zero eigenvalues, at least two of which are distinct.  
\end{claim} 
 
\begin{proof} 
The $q$ column vectors in $\calB'$ are linearly independent, so $rank(\calB') = rank(\calB'\calB'^T) = q$. Since $\calB'\calB'^T$ is square, it has $q$ non-zero eigenvalues.  
It is sufficient to show that at least two of those eigenvalues, say $\lambda_1$ and $\lambda_2$, are unequal.  
\eat{We assume that $\calB' \calB'^T$ has distinct columns, since otherwise, it merely corresponds to the same set being chosen multiple times.}  
As $\calB'$ does not correspond to an exact 3-cover, at least one row has more than one 1, and so at least one row is all 0's. The corresponding row and column in $\calB' \calB'^T$ will also be all 0's. 
 
Define the weighted graph induced by $\calB' \calB'^T$ as $G = (V, E, w)$ such that $|V| = k$,  
\eat{$(\calB' \calB'^T)_{ij} = \calB'_{*i} \cdot \calB'^T_{*j} \geq 1 \Rightarrow e = (i, j) \in E$,}  
$w(i,j) = (\calB' \calB'^T)_{ij}, \forall i, j \in [k]$. The all-zero rows correspond to isolated nodes. We know that the eigenvalues of the the matrix $\calB'\calB'^T$ are identical to those of the induced graph $G$, which in turn are the same as those of the connected components of $G$. Consider a non-isolated node $i$. Since each row of $\calB'$ is non-orthogonal to at least two other rows, it follows that $(\calB' \calB'^T)_{ij} \geq 1$ for at least 2 values of $j \neq i$. Thus, each non-isolated node is part of a connected component of size $\geq 3$ and since there are isolated nodes, the number of (non-isolated) components is $<q$. Thus, the $q$ non-zero eigenvalues of $G$ are divided among the $<q$ components of $G$.  

By the pigeonhole principle, there is at least one connected component with $\geq 2$ eigenvalues, call them $\lambda_1, \lambda_2$, say $\lambda_1\ge \lambda_2$. We know that a component's largest eigenvalue has multiplicity $1$, from which it follows that $\lambda_1\ne \lambda_2$, as was to be shown.  
\end{proof}

We next establish two helper lemmas, where $\calM$ denotes a $k\times k$ symmetric matrix. 

%%%%%%%%%%%%%%%%%%  
\begin{lemma} 
\label{lemma:1} 
Let $\calM$ be a positive semidefinite matrix \cite{horn2012matrix} of rank $q$. Suppose that it can be expressed as a sum of rank one matrices, i.e., $\calM = \sum_{i=1}^q {\bf a_i \cdot a_i}^T$, where $\bf a_i$ is a column vector, and $\forall i, j \in [k], i \neq j, {\bf a_i \cdot a_j}^T = 0, \mbox{ and } {\bf a_i \cdot a_i}^T = s$  . %and $||{\bf a_i}|| = ||{\bf a_j}|| = s$.  
Then the $q$  eigenvalues of $\calM$ are identical and equal to $s$. 
\end{lemma} 
 
\begin{proof} 
The spectral decomposition of a rank $q$  matrix $\calM$ is given as 
\begin{align} 
\calM  = \sum_{i=1}^{q} \lambda_i \bf{u_i u_i}^T \label{eq:svd} 
\end{align} 
 
where $\lambda_i$ are eigenvalues and $\bf u_i$ are orthonormal vectors. From the hypothesis of the lemma, we have $\frac{1}{s} \cdot \bf{a_i \cdot a_i}^T = 1$. %Setting $X_i' = \frac{X_i}{\sqrt{s}}$, 
\begin{align} 
\calM & = \sum_{i=1}^q {\bf a_i a_i}^T = \sum_{i=1}^q s \times \frac{\bf a_i}{\sqrt{s}} \frac{\bf a_i}{\sqrt{s}}^T 
\end{align} 
where $\frac{\bf a_i}{\sqrt{s}}$ are orthonormal. Comparing this with Eq.  \ref{eq:svd}, the eigenvalues of $\calM$ are $\lambda_1 = \lambda_2 = ... = \lambda_q = s$. 
\end{proof}
 
\begin{lemma} 
\label{lemma:2} 
Let $\calM$ be a symmetric rank $k$ matrix and suppose that it can be decomposed into $\sum_{i=1}^q {\bf a_i a_i}^T + \kappa\cdot I$, for some  constant $\kappa$. Then it has $(k-q)$ eigenvalues equal to $\kappa$. 
\end{lemma} 
 
\begin{proof} 
Let the eigenvalues of $\calM$ be  $\lambda_1, \lambda_2, ..., \lambda_k$. Let $\lambda$ any eigenvalue of $\calM$, and $v$ the corresponding eigenvector. Then we have $(\calM+\kappa I)v =  (\sum_{i=1}^q {\bf a_ia_i^T} +\kappa I)v = (\lambda + \kappa)v$.  Since $\sum_{i=1}^q {\bf a_i a_i}^T$ results in a rank $q$ symmetric matrix, it has $q$ non-zero eigenvalues. Adding $\kappa$ to all of them, we get, $\lambda_{q+1} = ... = \lambda_k = \kappa$. 
\end{proof}

\noindent 
{\bf Proof of Claim~\ref{claim1}:} 
Consider any set of vectors $\calB \subset \calW$: $|\calB| = k+q$. By Claim~\ref{claim3}, we may assume w.l.o.g. that $\calB$ includes all $k$ dummy vectors. Suppose $\calB' := \calB-\calD$ encodes an exact 3-cover of $X$. Then $\calB\calB^T$ can be decomposed into the sum of $q$ rank one matrices and a diagonal matrix: $\calB\calB^T = \sum_{j=1}^q {\bf b_j\cdot b_j^T} + \eta^2 I$. Here ${\bf b_i}$ refers to the $i$th column of $\calB$, which is a set vector.  Since $\calB'$ is an exact 3-cover, we further have that ${\bf b_i\cdot b_i^T} = 3$, $i\in [q]$, and ${\bf b_i\cdot b_j^T} = 0$, $i\neq j$. By Lemma~\ref{lemma:1}, since $\calB'\calB'^T$ is also a positive semidefinite matrix of rank $q$, we have $\lambdabprime_1 = \cdots \lambdabprime_q = 3$, where $\lambdabprime_i$ are the eigevalues of $\calB'\calB'^T$. The corresponding $q$ eigenvalues of $\calB\calB^T$ are all $\eta^2 + 3$. Furthermore, by Lemma~\ref{lemma:2}, the remaining $k-q$ eigenvalues of $\calB\calB^T$ are all equal to $\eta^2$. That is, the eigenvalues of $\calB\calB^T$ are $\lambdab_1 = \cdots = \lambdab_q = \eta^2 + 3$ and $\lambdab_{q+1} = \cdots = \lambdab_k = \eta^2$. For this $\calB$, $f(\calB) = tr((\calB\calB^T)^{-1}) = \frac{q}{\eta^2+3} + \frac{k-q}{\eta^2} = \theta$ (see Eq.~\ref{eq-theta}). 

Now, consider a set of vectors $\calA \subset \calW$, with  $|\calA| = k+q$, such that that $\calA$ includes all $k$ dummy vectors. Suppose $\calA' := \calA-\calD$ does not correspond to an exact 3-cover of $X$. Notice that $\calA$ is a symmetric rank $k$ matrix which can be decomposed into $\calA = \sum_{j=1}^q {\bf a_i\cdot a_i^T} + \eta^2 I$, so $\lambdaa_{q+1} = \cdots = \lambdaa_k = \eta^2$, where $\lambdaa_i$, $i\in[q+1,k]$, are $k-q$ of the eigenvalues of $\calA\calA^T$. Since both $\calB$ and $\calA$ include all $k$ dummy vectors and $q$ of the set vectors, by Claim~\ref{claim2}, $tr(\calB\calB^T) = tr(\calA\calA^T) = k + k\eta^2$. We have 
$\sum_{j=q+1}^q \lambdab_j = (k-q) \eta^2 = \sum_{j=q+1}^k \lambdaa_j$ and so $\sum_{j=1}^q \lambdaa_j = \sum_{j=1}^q \lambdab_j = q(\eta^2+3)$. Now, $f(\calB) = \sum_{j=1}^k \frac{1}{\lambdab_j} = \frac{q}{\eta^2+3} + \frac{k-q}{\eta^2}$, whereas $f(\calA) = \sum_{j=1}^k \frac{1}{\lambdab_j} = \sum_{j=1}^q \frac{1}{\lambdaa_j} + \frac{k-q}{\eta^2}$. Thus, to show that $f(\calB) < f(\calA)$, it suffices to show that $\frac{q}{\eta^2 +3} < \sum_{j=1}^q \frac{1}{\lambdaa_j}$. $\mbox{LHS } = q\times \frac{1}{AM(\lambdab_1, ..., \lambdab_q)} = q\times \frac{1}{AM(\lambdaa_1, ..., \lambdaa_q)}$, where $AM(.)$ denotes the arithmetic mean. $\mbox{RHS } = q\times \frac{1}{HM(\lambdaa_1, ..., \lambdaa_q)}$, where $HM(.)$ denotes the harmonic mean. It is well known that $AM(.) \ge HM(.)$ for a given collection of positive real numbers and the equality holds iff all numbers in the collection are identical. On the other hand, we know that since $\calA'$ does not correspond to an exact 3-cover of $X$, by Claim~\ref{claim4}, not all eigenvalues of $\calA'$ are equal, from which it follows that $LHS < RHS$, completing the proof of Claim~\ref{claim1} as also Theorem~\ref{thm-oid}. $\qed$ \\

%%%%%%%%%%%%%
We next establish an inapproximability result for OID. 

\eat{
\begin{theorem} 
\label{thm-m3c-approx}  
The Maximum q-Cover by 3-Sets problem is NP-hard and APX-Complete.  
\end{theorem}

\begin{proof}
Suppose not. Then it means that a PTAS exists for M3C. We can use that to solve the $k$-Set Cover problem for $k = 3$ in the following manner: For $ q= 1, ..., |S|$, run the PTAS for M3C until $\calC = |U|$. This gives us a PTAS for the $k$-Set Cover problem. However, it is APX-Complete; such a PTAS cannot exist unless P = NP. Hence the PTAS for M3C must also not exist, which implies that it is also APX-Complete.
\end{proof}
Next we use a gap-preserving reduction from M3C, to prove the following theorem. %Let $|U| = 3q$.

}

\subsection{Hardness of Approximation}
\label{sec:appx}
\begin{theorem} 
\label{thm-oid-approx}  
It is NP-hard to approximate the OID problem (Problem ~\ref{probnew-oid})  within a factor less than $\frac{\alpha}{\theta}$, where $\alpha = \theta + \frac{2}{(2 + \eta^2)(4 + \eta^2)(3 + \eta^2)}$.  
\end{theorem} 

%\begin{proof}

First we define a variant of the X3C problem 
%the well-known Maximum-Cover-by-k-sets problem~\cite{garey1979}, 
which we refer to as Max $q$-Cover by 3-Sets (M3C), which will be convenient in our proof. 

\begin{definition}
Given a number $q$ and a collection of sets $\sS = \{S_1, S_2, ..., S_n\}$, each of size 3, is there a subset $\sS^*$ of $\sS$ such that the cover $\calC = |\bigcup_{s \in \sS^*} s| = 3q$ and $| \sS^*| \leq q$?
\end{definition}

Since each set has 3 elements, with $| \sS^*| \leq q$, we get $\calC = 3q$ if and only if $| \sS^*|$ is an exact cover. Thus X3C can be reduced to M3C, making M3C NP-hard. 

We convert an instance $x$ of M3C to an instance of OID, $h(x)$, in the same way as described in the NP-Hardness proof: let the set of items be $\calI = \{a_1, ..., a_n, d_1, ..., d_k\}$, where $k=3q$, item $a_j$  corresponds to set $S_j$,  $j\in[n]$, and $d_j$ are dummy items,  $j\in[k]$. Let the dummy vectors be defined as above, and $b := q + k$. As shown previously in Claim \ref{claim3}, we need to only consider those sets of vectors $\calB$ that have all $k$ dummy vectors. Similarly, we can transform a solution $y$ of OID, back to a solution of M3C, $g(y)$, in the following manner: discard the chosen dummy vectors, and take the sets corresponding to the $q$ set vectors.

As a YES instances of M3C corresponds to a YES instances of X3C, an instance $x$ with $\calC = 3q$ corresponds to $f(\calB) = \theta$. 

For  the NO instances of M3C, $\calC \leq 3q-1$ (by the definition). Unfortunately, a similar one-to-one mapping does not exist in such cases: with the same $\calC$, there could be multiple instances of M3C that correspond to different instances of OID and correspondingly $f(\calB)$. From Theorem \ref{thm-oid}, we know that it is NP-hard to determine whether $f(\calB) \leq \theta$ for a given instance of OID -- $h(x)$.

To find the lowest $f(\calB)$ of a NO instance of OID, we first use an intermediate result that shows that among the set of different $f(\calB)$ values giving the same cover value $\calC$, the lowest possible $f(.)$ value increases as $\calC$ decreases.
%the lowest possible $f(\calB)$ for a given $C$, increases as $C$ decreases.

\begin{claim}
As the cover value increases,  the best (i.e., lowest) f(.) value among all the solutions with the same cover value decreases. 
%For any instance of M3C, among all the solutions that cover the same number of elements, $C_1$, the lowest corresponding $f(.)$ will be greater than the lowest corresponding $f(.)$ among solutions that cover $C_2$ elements, if and only if $C_1 < C_2$.
%For any instance of M3C, if there are two possible solutions $g(y_1), g(y_2)$ with $C$ values $C_1, C_2$, then the lowest possible $f(.)$ values for $y_1, y_2$ 
\end{claim}

\begin{proof}
Let $\calB' = \calB \setminus \calD$.

By interpreting $\calB' \calB'^T$ as a $(k \times k)$ adjacency matrix, the dimensions correspond to the $k$ nodes in the graph. Dimensions that are uncovered are isolated nodes, and dimensions that are covered are part of a connected component. Sum of degrees of the entire graph = $3k$ (sum of all entries in the adjacency matrix $\calB' \calB'^T$) which is a constant given $k$.

From this, given that the sum of the degrees over the graph is $3k$ (which is a constant), we argue that with more uncovered dimensions/nodes, average degree ($d_{avg}$)
(ignoring the isolated nodes) and maximum degree ($d_{max}$) increase. 
From this, it follows that each non-isolated node has degree at least $3$, hence the average degree for such nodes is greater than $3$ for any $\calB' \calB'^T$.
If there are multiple components in a given graph, considering the one with the highest average degree, $\lambda_1 \geq max(d_{avg}^{max}, \sqrt{d_{max}})$, where $d_{avg}^{max}$ is the highest average degree among all components. 
%For a YES instance, each node would have degree = $d_{avg}$ = $d_{max}$ = 3. 

For a NO instance, the highest average degree among all connected components is greater than $3$, since the vectors must overlap at least over $1$ dimension. For a given cover $\calC$, the lowest value of $\lambda_1$ is thus lower bounded by $d_{avg} > 3$, which increases as the  overlap increases. In turn, a higher value of $\lambda_1$ makes the distribution of eigenvalues more skewed, leading to a higher $f(.)$. To have a lower $f(.)$, we must have $\lambda_1$ as close to $3$ as possible, by  decreasing $d_{avg}$ and $d_{max}$, thereby, increasing coverage. 
%This relation finds a maxima for $C$ (correspondingly, a minima for $f(.)$) in case of a YES instance. \eat{In a NO instance, the highest coverage is $k-1$.}
\end{proof}

Following this claim, among the NO instances of OID, it is sufficient to show that the lowest $f(.)$ corresponds to the highest $\calC$, where $\calC = 3q-1$. Next we calculate its corresponding $f(.)$, and moreover, show that for a NO instance with $\calC = 3q-1$, we get a unique OID solution. For this scenario, it can be shown that there are exactly $q-2$ disjoint sets, and $2$ sets cover exactly one element twice. This can only be obtained from a solution $y$ of OID, if in the given solution, $q-2$ set vectors are disjoint, and $2$ have exactly one $1$ in the same position. The following example illustrates this.

\begin{example}\label{ex1} 
For an instance with $q=3$, a solution with exactly two vectors overlapping on one dimension could look like 

\[
   \calB'^T=
   \begin{bmatrix}
   1 & 1 & 1 & 0 & 0 & 0 & 0 & 0 & 0 \\
   0 & 0 & 1 & 1 & 1 & 0 & 0 & 0 & 0 \\     
   0 & 0 & 0 & 0 & 0 & 1 & 0 & 1 & 1
	\end{bmatrix}
\]

Then 
\[
   \calB'\calB'^T=
   \begin{bmatrix}
   1 & 1 & 1 & 0 & 0 & 0 & 0 & 0 & 0 \\
   1 & 1 & 1 & 0 & 0 & 0 & 0 & 0 & 0 \\
   1 & 1 & 2 & 1 & 1 & 0 & 0 & 0 & 0 \\
   0 & 0 & 1 & 1 & 1 & 0 & 0 & 0 & 0 \\     
   0 & 0 & 1 & 1 & 1 & 0 & 0 & 0 & 0 \\
   0 & 0 & 0 & 0 & 0 & 1 & 0 & 1 & 1 \\
   0 & 0 & 0 & 0 & 0 & 1 & 0 & 1 & 1 \\
   0 & 0 & 0 & 0 & 0 & 0 & 0 & 0 & 0 \\
   0 & 0 & 0 & 0 & 0 & 1 & 0 & 1 & 1
	\end{bmatrix}
\]

Next, we show what $f(\calB)$ of such a solution would be. As before, let $\calB' := \calB \setminus \calD$. Interpreting $\calB' \calB'^T$ as the adjacency matrix of a graph $G$, we know that the eigenvalues of $\calB' \calB'^T$ are the same as those of $G$, given by the multi-set union of its components, which are: $q-2$ corresponding to the disjoint set vectors, and $1$ corresponding to the two over-lapping vectors. The first $q-2$ components each form a 3-regular graph which contributes an eigenvalue of $3$ each. It could be shown that the last one, which corresponds to the overlap, contributes to $(0, 0, 0, 0, 2, 4)$. Therefore, $f(\calB) = \alpha = \frac{2q}{\eta^2} + \frac{q-2}{3 + \eta^2} + \frac{1}{2 + \eta^2} + \frac{1}{4 + \eta^2} = \theta + \frac{2}{(2 + \eta^2)(4 + \eta^2)(3 + \eta^2)}$. \qed 
%If more set vectors were to overlap, it would lead to a higher skew in the distribution of the trace into eigenvalues, and thus a higher value of $f(.)$.
\end{example} 

It follows from our arguments, that $f(\calB) \geq \alpha$ if and only if $\calC \leq 3q-1$. 

Let $\calA$ be an approximation algorithm that approximates OID to within $c < \frac{\alpha}{\theta}$, returns a value $v$ such that $OPT_{OID}(h(x)) \leq v \leq c \times OPT_{OID}(h(x))$.

\begin{claim}
$x$ is a YES instance of M3C if and only if $\theta \leq v < \alpha$.
\end{claim}

\begin{proof}
If $h(x)$ is a YES instance of OID, $OPT_{OID}(h(x)) = \theta$, so $\theta \leq v \leq c \times \theta$. Since $c < \frac{\alpha}{\theta}$, $\theta \leq v < \alpha$. If $h(x)$ is a NO instance of OID, $\alpha \leq OPT_{OID}(h(x))$, so $\alpha \leq v$. Since the intervals are disjoint, the claim follows.
\end{proof}

\noindent 
{\bf Proof of Theorem~\ref{thm-oid-approx}}:
Finally,  if such an approximation algorithm $\calA$ existed, we would be able to distinguish between the YES and NO instances of M3C in polynomial time. However as that is NP-hard, unless P = NP, $\calA$ cannot exist.

\subsection{Supermodularity and Submodularity}
\label{sec:supermodularity}
\label{sec:nosubsuper} 

If the objective function were to satisfy the nice property of submodularity or supermodularity, we could exploit it to devise some approximation algorithm. 
First we review the definitions of submodularity and supermodularity. 
\begin{definition}
For subsets $A \subset B \subset U$ of some ground set $U$, and $x \in U \setminus B$, a set function $F: 2^U \rightarrow \mathbb{R}_{\ge 0}$ is submodular if $F(B \cup \{x\}) - F(B) \leq F(A \cup \{x\}) - F(A)$. The 
function $F(.)$ is supermodular iff $-F(.)$ is submodular, or equivalently iff $F(B \cup \{x\}) - F(B) \geq F(A \cup \{x\}) - F(A)$.
\end{definition}

In \cite{anava2015budget}, for a similar objective function for the user selection problem for a cold-start item, the authors claimed that their objective function is supermodular. The following lemma shows that the objective function $f(.)$ for our OID problem is not supermodular. 

\begin{lemma}
The objective function $f(V_B C_B) = tr((V_B C_B^{-2} V_B^T)^{-1})$ of the OID problem is {\em not supermodular}.
\end{lemma}
\begin{proof}
We prove the result by showing that the function $f(\calM) =  tr(\calM\calM^T)^{-1})$ is in general not supermodular. 
Consider the following matrices:  

\[
\calM_1 =
  \begin{bmatrix}
    1 & 1 & 0 & 0 & 0 \\ 
    0 & 0 & 1 & 0 & 0 \\ 
    1 & 0 & 0 & 1 & 0 \\ 
    0 & 0 & 1 & 1 & 1 \\ 
    0 & 0 & 1 & 0 & 1 
  \end{bmatrix},\quad
\calM_2 =
  \begin{bmatrix}
    0 & 0 & 1 & 0 & 1 & 1 \\ 
    1 & 0 & 0 & 0 & 0 & 1 \\ 
    0 & 1 & 0 & 0 & 1 & 1 \\ 
    1 & 1 & 0 & 1 & 0 & 0 \\ 
    1 & 0 & 0 & 1 & 0 & 1 
  \end{bmatrix}, \quad
\]
and vector 
\[
x^T=
  \begin{bmatrix}
    0 & 1 & 0 & 0 & 0
  \end{bmatrix}
\] 

Notice that $\calM_1$, viewed as a set of column vectors, is a subset of $\calM_2$, viewed as a subset of column vectors. Now, 
$f(\calM_1) = tr(\calM_1\calM_1^T) =12, f(\calM_1\cup \{x\}) = 10.333, f(\calM_2) = 6.6250, f(\calM_2 \cup \{x\}) = 4.4783$.

Clearly, $f(\calM_1 \cup \{x\}) - f(\calM_1) = 10.333 - 12 = -1.6667$ 
and $f(\calM_2 \cup \{x\}) - f(\calM_2) = 4.4783 - 6.6250 = -2.1467$, 
which violates $f(\calM_1 \cup \{x\}) - f(\calM_1) \le 
f(\calM_2 \cup \{x\}) - f(\calM_2)$, showing $f(.)$ is not supermodular. 
\end{proof}

We remark that the lack of supermodularity of $f(.)$ is not exclusive to binary matrices; supermodularity does not hold for real-valued matrices $\calM$ as well. As a consequence, by the duality between the technical problems of cold-start users and cold-start items, the lemma above disproves the claim in \cite{anava2015budget} about the supermodularity of their objective function. Similarly, one can show that the objective function is also not submodular. These results together with Theorem~\ref{thm-oid-approx} dash hopes for finding approximation algorithms for the OID problem. 

\eat{
\begin{lemma}
The objective function of $f(B) = tr((V_BV_B^T)^{-1})$ the OID problem is {\em not sub-modular}.
\end{lemma}
\begin{proof}

Consider the same matrix $\calM_2$ and $x$ as above and the following matrix $\calM_1$. 

 \[
\calM_1=
  \begin{bmatrix}
    0 & 0 & 0 & 1 & 1 \\ 
    1 & 0 & 0 & 0 & 1 \\ 
    0 & 1 & 0 & 1 & 1 \\ 
    1 & 1 & 1 & 0 & 0 \\ 
    1 & 0 & 1 & 0 & 1 
  \end{bmatrix}
\] 
As before, regarded as sets of column vectors, $\calM_1 \subset \calM_2$. However, $f(\calM_1) = 20, f(\calM_1 \cup \{x\}) = 16.333, f(\calM_2) = 6.6250, f(\calM_2 \cup \{x\}) = 4.4783$. Hence $f(\calM_1 \cup \{x\}) - f(\calM_1) = 16.333 - 20 = -3.6667$ 
and 
$f(\calM_2 \cup \{x\}) - f(\calM_2) = 4.4783 - 6.6250 = -2.1467$. Clearly, submodularity fails to hold, since $f(\calM_1 \cup \{x\}) - f(\calM_1) \geq f(\calM_2 \cup \{x\}) - f(\calM_2)$ is violated. 
\end{proof}
}

\section{Algorithms}
\label{sec:approach}
%\label{sec:algos} 

%In this section, we present algorithms for selecting items with which to interview a cold user so as to learn her profile as accurately as possible. 
In view of the hardness and hardness of approximation results (Theorems~\ref{thm-oid} and~\ref{thm-oid-approx}), and the fact that the objective function is neither submodular nor supermodular (see Section~\ref{sec:nosubsuper}), efficient approximation algorithms are unlikely to exist. We present scalable heuristic algorithms for selecting items with which to interview a cold user so as to learn her profile as accurately as possible. %item selection. 

\subsection{Forward Greedy} 

%As discussed earlier, in \cite{anava2015budget}, the authors study the problem of user selection for the cold-start item problem. They propose two backward greedy algorithms, called Backward Greeedy (BG) and Backward Greedy2 (BG2). 

Recall that for a set of items $S \subseteq \{v_1, ..., v_n\}$, we denote by $V_S$ the submatrix of the item latent factor matrix corresponding to the items in $S$, $C_S$ is a diagonal matrix with $\sigma_1, \sigma_2, ..., \sigma_{|S|}$ at positions corresponding to items in $S$ and $f(V_SC_S) := tr((V_S C_S^{-2} V_S^T)^{-1})$ is the profile learning error that we seek to minimize by selecting the best items. It can be shown, following a similar result in \cite{anava2015budget} (Proposition 3) that $f(.)$ is monotone decreasing, i.e., for item sets $S \subseteq T$, $f(V_TC_T) \leq f(V_SC_S)$, and so $-f(.)$ is monotone increasing.
 
\noindent 

{\bf Overview}. We start with $B$ initialized to the empty set of items, and in the first iteration, add the item that leads to the smallest expected error value, i.e., smallest value of $f(.)$.
%which has the highest error and hence $-f(\emptyset)$ has the smallest value. 
Then, in each successive iteration, we add to $B$ an item that has the maximum marginal gain w.r.t. $-f(.)$. That is, we successively add 
\begin{align*} 
v^* = \mbox{arg max}_{v \in \calI\setminus B} [-f(V_{B\cup\{v\}} C_{B\cup\{v\}}) -(-f(V_B C_B))]  \\ 
= \mbox{arg max}_{v \in \calI\setminus B} [f(V_B C_B) - f(V_{B\cup\{v\}} C_{B\cup\{v\}})]
\end{align*} 
to $B$ until the budget $b$ is reached. We use $-\ftilde(v_j|V_B)$ to denote $[f(V_B C_B) - f(V_{B\cup\{v\}} C_{B\cup\{v\}})]$.

%start with the set of all items and successively remove an item with the smallest increase in the error, until no more than $b$ items are left, where $b$ is the budget. It was claimed in \cite{anava2015budget} that the backward greedy algorithms are approximation algorithms. This is incorrect since their claim relies on the error function being supermodular. Unfortunately, their proof of supermodularity is incorrect as shown by our counterexample in Section~\ref{sec:nosubsuper}. Thus, backward greedy is a heuristic for their problem as well as our OID problem. 

\noindent 
{\bf Acceleration}. 
In this section, we propose an accelerated version of Forward Greedy (FG), by borrowing ideas from the classic lazy evaluation approach, originally proposed in~\cite{minoux1978accelerated} to speed up the greedy algorithm for submodular function maximization.

Recall that our error function $f(.)$ is actually \emph{not} supermodular, and hence $-f(.)$ is also not submodular. Our main goal in applying lazy evaluation to it is not only to accelerate item selection, but also explore the impact of lazy evaluation on the error performance. It allows us to save on evaluations of error increments that are deemed redundant, assuming (pretending, to be more precise) that $f(.)$ is supermodular. We will evaluate both the prediction and profile error performance as well as the running time performance of these optimizations in Section~\ref{sec:expt}. 

A further speed-up can be obtained by using the Sherman-Morrison optimization which saves on repeated invocations of matrix inverse, and instead computes incrementally and hence efficiently, using rank one updates \cite{sherman1949adjustment}.

%We give the pseudocode for the accelerated version of BG2 in Algorithm~\ref{alg:absg2}. The algorithm corresponding to accelerated BG can be easily obtained from this, as we explain below. 

%We start with $B$ initialized to all items. We use a priority queue for efficient implementation of lazy evaluation. The ``freshness'' check in Line \ref{line:mg} in Algorithm  \ref{alg:absg2} checks that the error increment in $\alpha_j$ was computed in the latest iteration to make sure $v_j$ is indeed the item with the least error increment. Notice that we use the notation $\ftilde(v_j | V_B, C_B)$ where $V_B$ is the latent factor matrix corresponding to the items in $B$ and $C_B$ is the covariance matrix. $C_B$ is initialized to the diagonal matrix of noise terms $C$, where $C_{jj} = \sigma_j^2$, $v_j\in B$. 

%The use of Sherman-Morrison optimization allows us to save on repeated invocation of matrix inverse, and instead allows it to be computed incrementally and hence efficiently using rank one update. %The use of lazy evaluation saves on evaluations of error increments that are deemed redundant, assuming (pretending, to be more precise) that $f(.)$ is supermodular. We will evaluate both the prediction and profile error performance as well as the running time performance of these optimizations in Section~\ref{sec:expt}. 

Apart from the algorithm described above, for the general Forward Greedy (FG2), we also study a more basic version (FG1) as a baseline, assuming that the noise terms are identical, i.e., $C =\sigma I$, where $I$ is the identity matrix. This saves some work compared to FG2. We refer to their accelerated versions as AFG1 and AFG2. 
%The accelerated version of basic Forward Greedy (FG) differs from Algorithm~\ref{alg:absg2} by simply assuming that the noise terms are identical, i.e., $\sigma_1^2 = \cdots = \sigma_n^2 = \sigma^2$, i.e., we set $C$ to $\sigma^2\cdot I$, where $I$ is the identity matrix. This saves some work compared to Algorithm~\ref{alg:absg2}. We refer to this modified algorithm as ABG1 and omit its pseudocode for brevity. 

\subsection{Backward Greedy}

We compare the FG family of algorithms against the backward greedy algorithms BG1 and BG2 proposed in \cite{anava2015budget}. It was claimed there that the backward greedy algorithms are approximation algorithms. This is incorrect since their claim relies on the error function being supermodular. Unfortunately, their proof of supermodularity is incorrect as shown by our counterexample in Section~\ref{sec:nosubsuper}. Thus, backward greedy is a heuristic for their problem as well as our OID problem.

\noindent 
{\bf Overview}. Backward greedy (BG) algorithms essentially remove the worst items from the set of all items and use the remaining ones as interview items. We start with the set of all items and successively remove an item with the smallest increase in the error, until no more than $b$ items are left, where $b$ is the budget. 

As with accelerated forward greedy, we use lazy evaluation to optimize backward greedy. We refer to the resulting algorithm as Accelerated Backward Greedy, and study the basic version (ABG1) with identical variances and the general version (ABG2).

One key shortcoming of the BG family of algorithms (BG1 and BG2 and their accelerated versions) is that they need to sift through all items in the database and eliminate them one by one till the budget $b$ is reached. In a real recommender system, the number $n$ of items may be in the millions and $b$ is typically $<< n$, so this approach may not be feasible to deploy in real world systems. 

In the next section, we conduct an empirical evaluation of the forward and backward greedy algorithms as well as their accelerated versions proposed here and compare them against baselines. 

\eat{
\begin{algorithm}
  \caption{Accelerated Forward Greedy  (AFG)
    \label{alg:ags} 
\label{alg:afg}}
  \begin{algorithmic}[1]
    \Require{item set $\calI$ and corresponding matrix $V$; budget $b$.}
    \Ensure{items subset $B, |B| = b$}
    \Statex
      \Let{$V_B$}{$\phi$}
      \Let{$B$}{$\phi$}
      \For{$j \gets 1 \textrm{ to } |\calI|$}
      \State insert $(v_j, -\ftilde(v_j|V_B))$ into priority queue $Q$ 
      \EndFor
      \Do
      \State pop $(j, \alpha_j)$ from $Q$ 
      \If{$\alpha_j$ not "fresh"}
      \State recompute $\alpha_j = - \ftilde(v_j | V_B)$
      \EndIf
      \If{$\alpha_j >$ MAX($Q$)}
	\Let{$V_B$}{$V_{B} \cup \{v_j\}$}
	\Let{$B$}{$\calI \cup \{v_j\}$}
      \Else
      	\State insert $(v_j, \alpha_j)$ into $Q$ 
      \EndIf
      \doWhile{$|V_B| < b$}
  \end{algorithmic}
\end{algorithm}
}

\section{Experimental Evaluation}
\label{sec:expt}
In this section, we describe the experimental evaluation, and compare with prior art. We evaluate our solutions both qualitatively and scalability-wise: for quality evaluation , we measure {\em prediction error} and {\em user profile estimation error} (see Section \ref{sec:setup}), whereas, for scalability, we measure the running time.

The development and experimentation environment uses a Linux Server with 2.93 GHz Intel Xeon X5570 machine with 98 GB of memory with OpenSUSE Leap OS. 

%All numbers  are presented as the average of three runs.

%We also discuss various implementation issues.

\subsection{Dataset and Model Parameters}
We use Netflix and Movielens (ML) datasets. For each dataset, we train a probabilistic matrix factorization model~\cite{pmf} on only the ratings given by the warm users. We use gradient descent algorithm~\cite{zhang2004solving} to train the model, with latent dimension = $20$, momentum = $0$, regularization = $0.1$ and linearly decreasing step size for faster convergence. This allows us to move quickly towards the minima initially, decreasing the step size as we get closer, to avoid overshooting. We report the dataset characteristics, and the RMSE obtained after training on the warm user ratings, in Table \ref{tab:datasets}. 
%We report the number of warm users, number of items rated by them, and RMSE obtained, for the different datasets in Table \ref{tab:expt}. 
%Moreover, we use three different available ML datasets that gives us a total of five different datasets. 
%We describe their characteristics in Table \ref{tab:datasets}.

\begin {table}[H]
\caption {Dataset Sizes} \label{tab:datasets} 
\begin{center}
%	\label{tab:datasets}
    \begin{tabular}{| l | c | c | c | c |}
    \hline
    \textbf{Dataset} & \textbf{\# Ratings} & \textbf{\# Users} & \textbf{\# Items} & \textbf{RMSE} \\ \hline
    ML 100K & 100,000 & 943 & 1682 & 0.9721 \\ \hline
    ML 1M & 1,000,209 & 6,040 & 3,900 & 0.8718 \\ \hline
    %Epinions & 664, 824 & 49,290 & 139,738 & 0.0097\% \\ \hline
    ML 20M & 20,000,263 & 138,493 & 27,278 & 0.7888\\ \hline
    Netflix & 100,480,507 & 480,189 & 17,770 & 0.8531\\ 
    \hline
    \end{tabular}
\end{center}
\end {table}

\eat{
\begin {table}[H]
\caption {Dataset Sizes} \label{tab:datasets} 
\begin{center}
%	\label{tab:datasets}
    \begin{tabular}{| l | c | c | c | c |}
    \hline
    \textbf{Dataset} & \textbf{\# Ratings} & \textbf{\# Users} & \textbf{\# Items} & \textbf{Sparsity} \\ \hline
    ML 100K & 100,000 & 943 & 1682 & 6.3\% \\ \hline
    ML 1M & 1,000,209 & 6,040 & 3,900 & 4.25\% \\ \hline
    %Epinions & 664, 824 & 49,290 & 139,738 & 0.0097\% \\ \hline
    ML 20M & 20,000,263 & 138,493 & 27,278 & 0.53\%\\ \hline
    Netflix & 100,480,507 & 480,189 & 17,770 & 1.18\%\\ 
    \hline
    \end{tabular}
\end{center}
\end {table}

\begin {table}[H]
\caption {Experiment Sizes} \label{tab:expt} 
\begin{center}
%	\label{tab:datasets}
    \begin{tabular}{| l | c | c | c | c |}
    \hline
    Dataset & \# Warm Users $|\calU|$ & $|\calI|$ & RMSE \\ \hline
    
    ML 100K & 702 & 1647 & 0.9721 \\ \hline
    ML 1M & 4,473 & 3,666 & 0.8718 \\ \hline
    %Epinions & 27,798 & 112,837 & 1.1788 \\ \hline
    ML 20M & 102,628 & 25,529 & 0.7888 \\ \hline
    Netflix & 355,757 & 17,770 & 0.8531 \\ 
    \hline
    \end{tabular}
\end{center}
\end {table}
}

\subsection{Experimental Setup}

%{\bf Experimental Setup:}
\label{sec:setup}
We simulate the cold user interview process as follows:
\begin{enumerate}
\item Set up the system
\begin{enumerate}
\item Randomly select 70\% of the users in a given dataset to train the model ($\calU$)
\item $R := $ Matrix of ratings given by $\calU$ only
\item Train a PMF model on $R$, to obtain $U, V$ \label{step:pmf}
\end{enumerate}
\item Construct item covariance matrix $C$ given by, \\ $\sigma_j := \sqrt{\frac{1}{|{R}_{*j}|} \sum_{i \in {R}_{*j}} (R_{ij} - \hat{R}_{ij})^2}$, where $R_{*j}$ refers to the column(set) of ratings received by item $j$ \label{step:cov}
\item For each cold user $u_{\ell} \not \in \calU$,
\begin{enumerate}
\item Construct $U_{\ell}$ using gradient descent method \cite{rendle2008online}, and using the item latent factor matrix $V$ \label{step:true-user-vector}
\item Randomly split items $u_{\ell}$ has rated, into candidate pool $CP$ and test set $Test$
\item Run item selection algorithm on $CP$ with corresponding $V$ and budget = $b$
\item $B :=$ items returned by algorithm to interview $u_{\ell}$
\item Reveal $R_{\ell}:= u_{\ell}$'s ratings on $B$
\item Construct $\uHat := (\gamma I + V_B C_B^{-2} V_B^T)^{-1} V_B C_B^{-2} R_{\ell}$
\item RMSE on $Test := \sqrt{\frac{1}{|Test|} \sum_{j \in Test} (R_{\ell j} - V_j^T \uHat)^2}$
\item Calculate profile error := $||\uHat - U_{\ell}||^2_F$
\end{enumerate}
%\item Average prediction and profile error over all cold users
\end{enumerate}

In Step \ref{step:cov}, we estimate the noise terms using the method outlined in \cite{anava2015budget}. For the case where the covariance matrix $C = \sigma I$, we estimate $\sigma$ that best fits a validation set (a randomly chosen subset of the cold user ratings).

In Step \ref{step:true-user-vector} we compute \emph{true} latent vectors of the cold users. We cannot compute that using the PMF model in Step \ref{step:pmf}, as these ratings are hidden at that stage. Moreover, since each cold user is independent, we cannot train a model on their combined pool of ratings. Instead, we adopt the gradient descent based method given in \cite{rendle2008online}, to generate the latent vector for a new user keeping everything else constant. This method produces results comparable to retraining the entire model globally (with 1\% error in the worst case), in a few milliseconds \cite{rendle2008online}. The user profiles thus generated are used as \emph{true profiles}.

\eat{we train the PMF model on ratings given by 70\% of the users only. The rest of the users are treated as cold users, whose ratings are kept \emph{hidden} from the item selection algorithms. %We only consider cold users who have provided more than 160 ratings, in order to ensure that we have sufficient items to train and test the user estimation model on. For the Netflix dataset, this is equivalent to more than 150,000 cold users, and for the Movielens datasets, this is equivalent to ... cold users. 
For each user, we randomly divide the pool of items for which we have their ratings, into a candidate pool and a test set. We run the algorithms to select items from the candidate pool, while keeping the corresponding ratings hidden. Then we \emph{reveal} their ratings on the items selected by the algorithms, as if obtained through the interview process. Based on the revealed ratings, we build the user profile, and use that to predict the cold users' ratings for the items in the test set. We use RMSE to evaluate the algorithms' performance.}

The estimated rating is computed by taking an inner product between the item and the cold user vector $u_{\ell}$. We call this the \emph{ideal} setting, as it corresponds to an ideal, zero error MF model. It has two advantages: first, it allows us to decouple our problem from the problem of tuning a matrix factorization model. This way, the model error from using a possibly less than perfect PMF model does not percolate to our problem. Second, we do not have to be limited to only selecting the items for which we have ratings in our database, as we can generate ratings for all items. 
%This is especially crucial for sparse datasets, like ours. 
This allows us to run scalability tests more comprehensively, as we are not limited by the availability of ratings. %This set up accurately indeed simulates a real world recommender system.

\subsection{Algorithms Compared}
We compare the following algorithms against their accelerated versions, as described in Section~\ref{sec:approach}: Backward Greedy Selection 1 ({\tt BG}), Backward Greedy Selection 2 ({\tt BG2}), Forward Greedy Selection 1 ({\tt FG1}) and Forward Greedy Selection 2 ({\tt FG2}). Further, we use the following heuristics as baselines: Popular Items ({\tt PI}), where $b$ items with the most ratings are selected, Random Selection ({\tt RS}), where the items are randomly sampled from the candidate pool, High Variance ({\tt HV}), where $b$ items with the highest variance in rating prediction among warm users are selected, Entropy ({\tt Ent}) and Entropy0 ({\tt Ent0}) \cite{rashid2008learning}, where the $b$ most contentious, and therefore most informative items are selected. Entropy0 is a modification of entropy that includes a notion of how many ratings have been received by each item, so as to discourage the selection of obscure items.
%This gives us a total of 8 algorithms to compare.

\subsection{Quality Experiments}

{\bf Results:} We run quality experiments to measure prediction error and profile error for all four datasets. For the datasets ML 100K and ML 1M, we compare all $13$ algorithms under the \emph{ideal} setting, where we note that {\tt BG} performs almost exactly the same as {\tt FG} (see Fig.~\ref{fig:quality}), while {\tt BG2} and {\tt FG2} perform better for both profile and prediction error.
For the larger datasets Netflix and ML 20M, we compare algorithms under the \emph{real} setting. For both, we observe that {\tt FG2} outperforms {\tt BG2} for both prediction and profile error, and {\tt FG} outperforms {\tt BG} for smaller values of $b$. 

%\note Senj can you please fix the referencing? Not sure why it's not working.
%Thus, for the rest of the experiments we focus only on the accelerated versions.

\begin{figure*}[ht]
\centering
\subfigure[Prediction Error -- ML 100K]{
   \includegraphics[height=3cm, width=4cm]{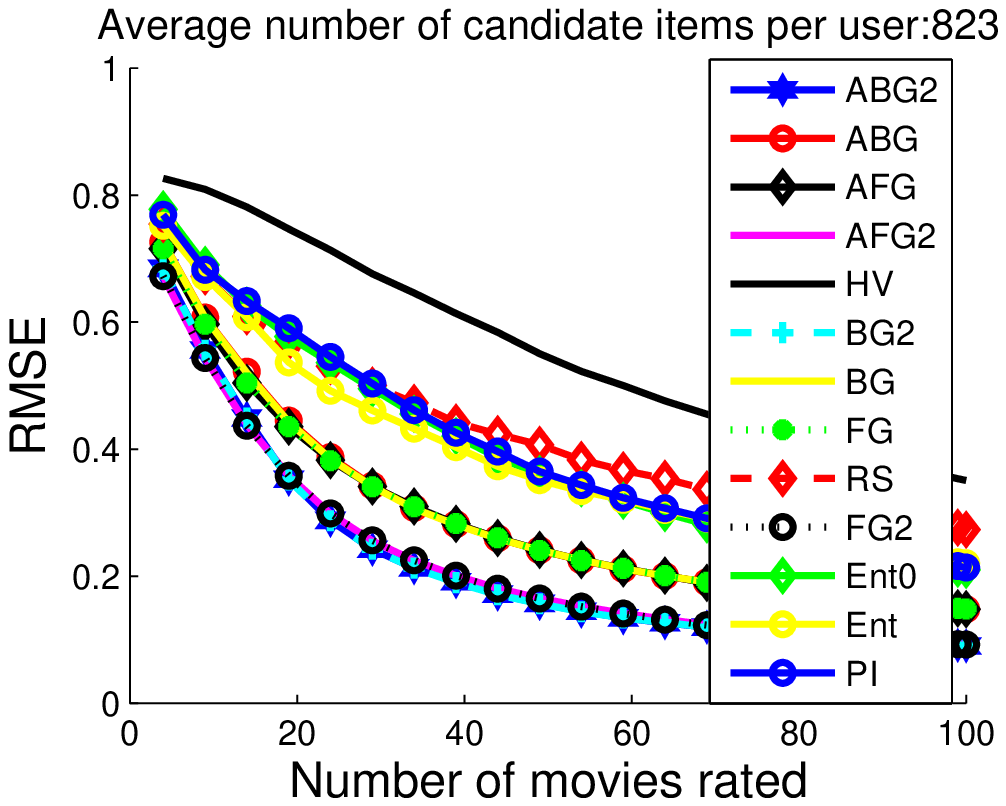}
    %\caption{Average Group Satisfaction on top-$k$ Itemset Varying Number of Users}
     \label{fig:ml100k_algos8}
    }
\subfigure[Profile Error -- ML 100K]{
   \includegraphics[height=3cm, width=4cm]
    {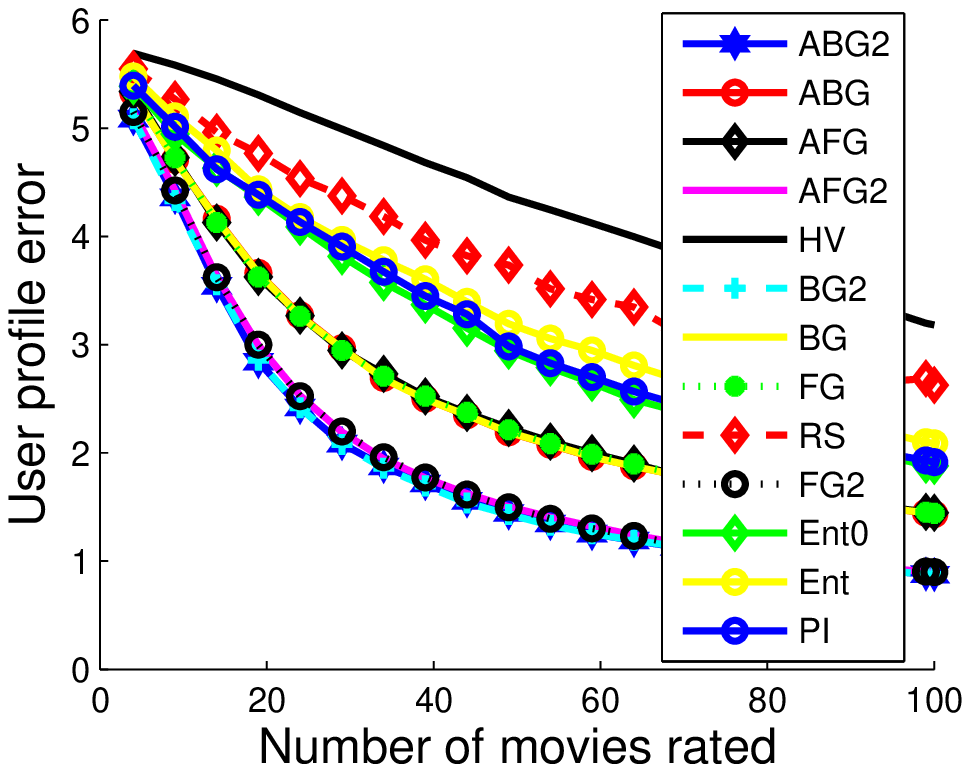}
   %\caption{Average Group Satisfaction on top-$k$ Itemset Varying Number of Items}
     \label{fig:ml100k_pu}
    }
\subfigure[Prediction Error -- ML 1M]{
   \includegraphics[height=3cm, width=4cm]{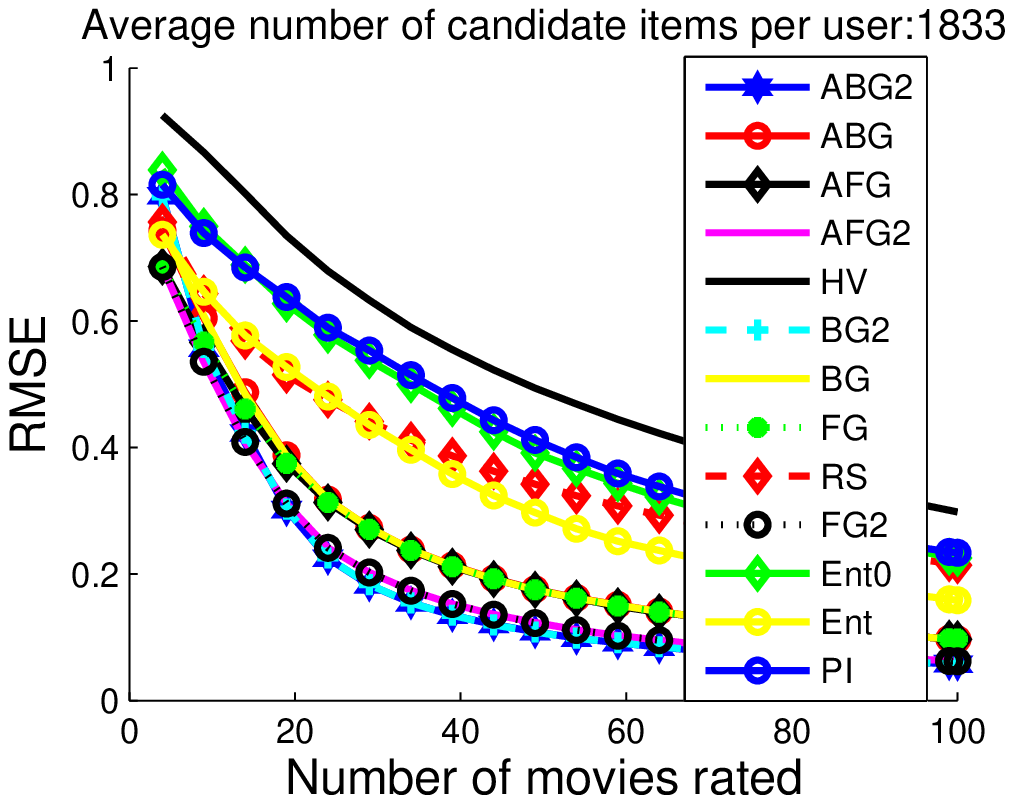}
    %\caption{Average Group Satisfaction on top-$k$ Itemset Varying Number of Groups}
     \label{fig:ml1m_algos8}
    }
\subfigure[Profile Error -- ML 1M]{
   \includegraphics[height=3cm, width=4cm]{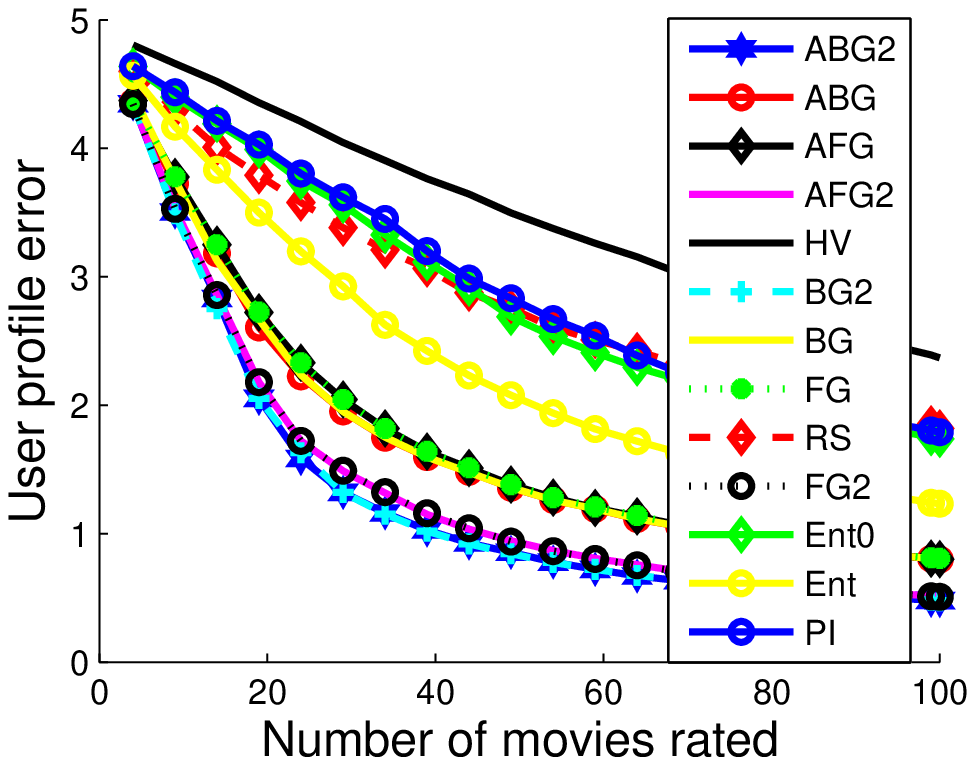}
    %synthetic/fig1WorkloadVsClockTime.pdf
    %\captionsat-av-user{Average Group Satisfaction on top-$k$ Itemset Varying Top-$k$}
    \label{fig:ml1m_pu}
    }
    \subfigure[Prediction Error -- Netflix]{
   \includegraphics[height=3cm, width=4cm]{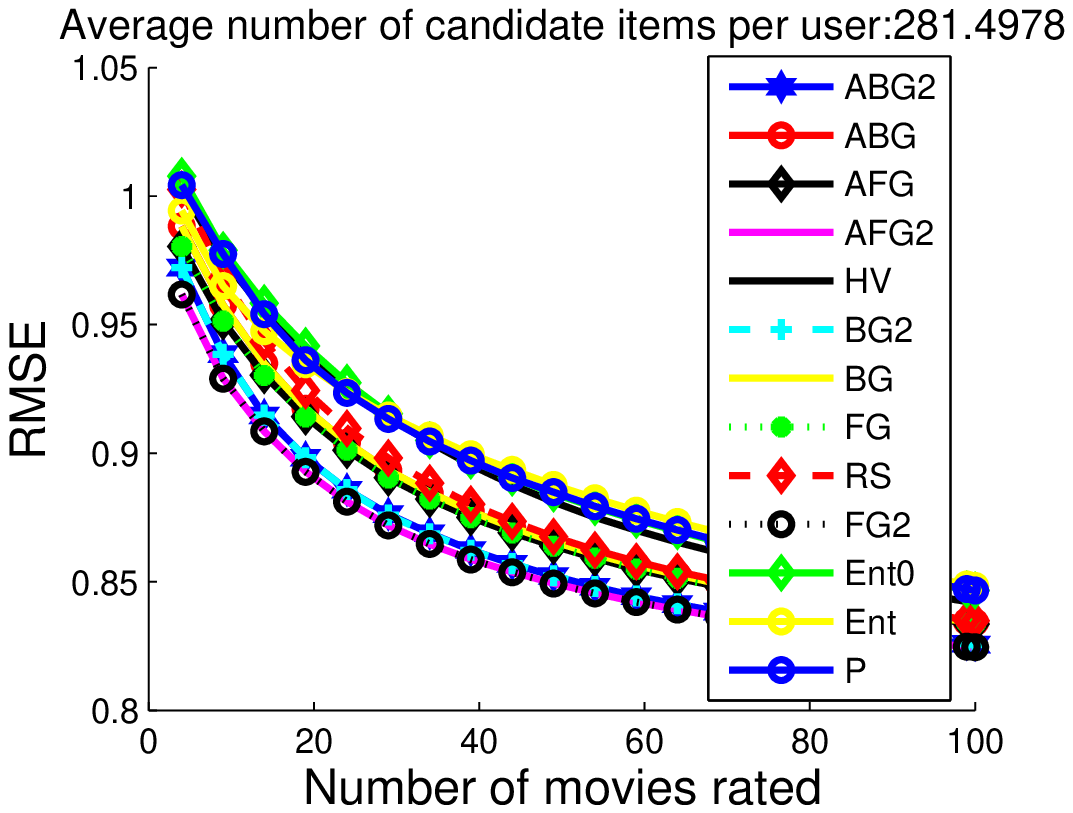}
     \label{fig:nf_algos8}
    }
\subfigure[Profile Error -- Netflix]{
   \includegraphics[height=3cm, width=4cm]
{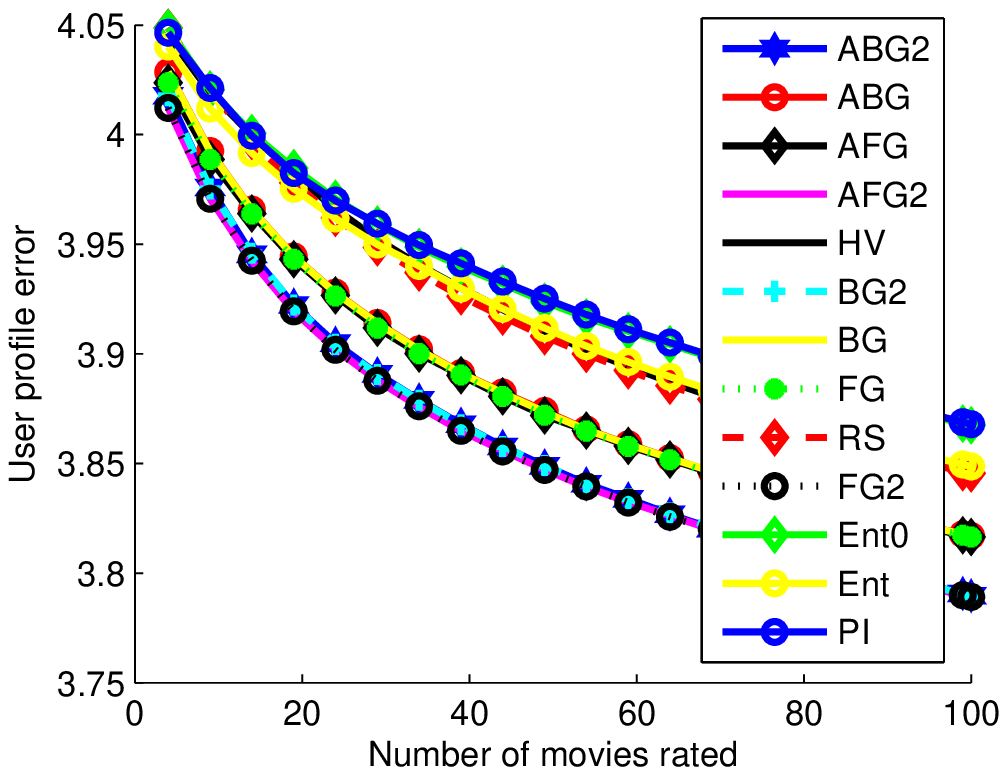}
     \label{fig:nf_pu}
    }
\subfigure[Prediction Error -- ML 20M]{
   \includegraphics[height=3cm, width=4cm]{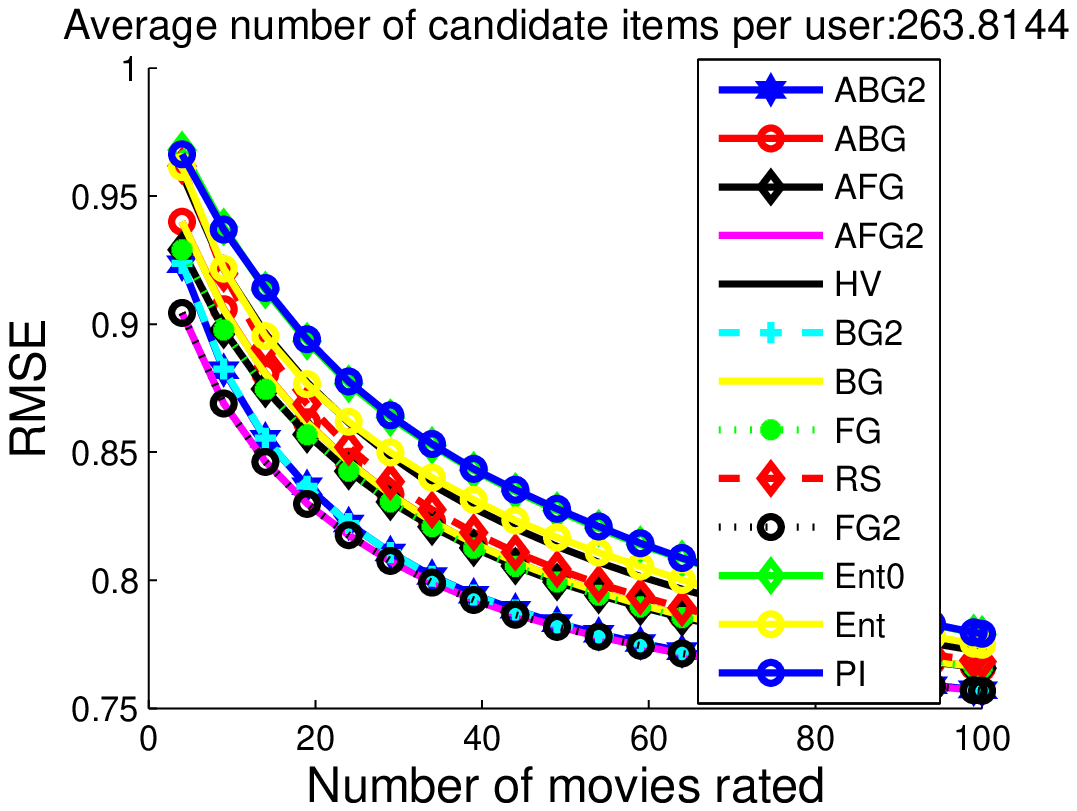}
  \label{fig:ml20m_algos8}
    }
\subfigure[Profile Error -- ML 20M]{
   \includegraphics[height=3cm, width=4cm]{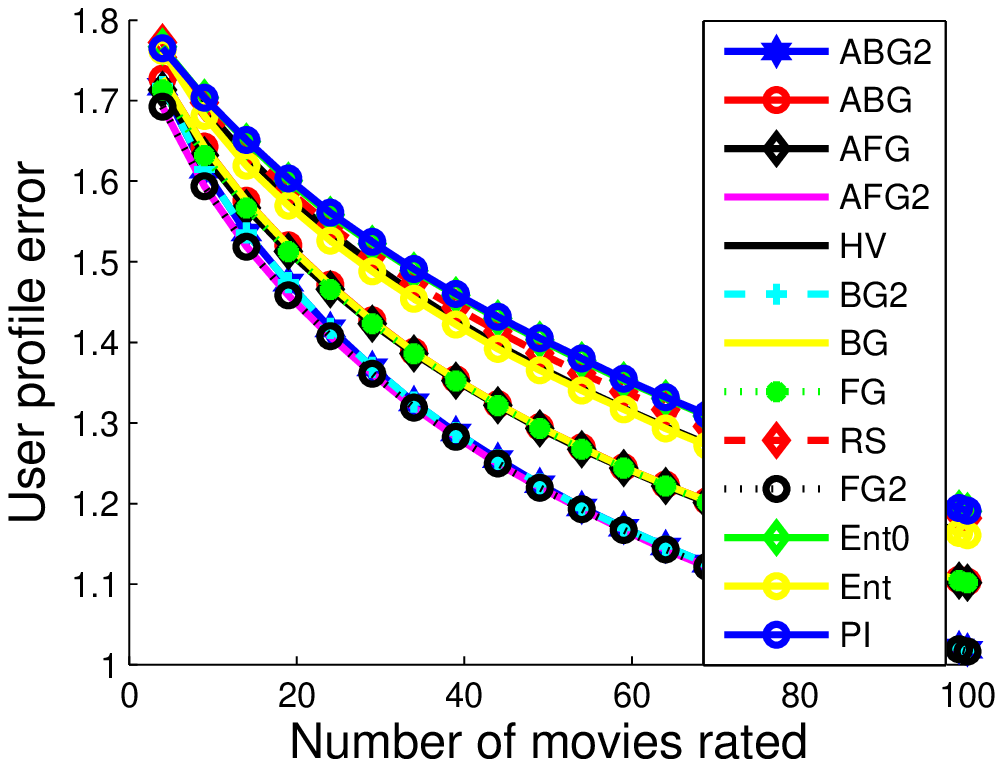}
  \label{fig:ml20m_pu}
    }
    \vspace{-0.1in}
\caption{\small Profile and Prediction Errors \label{fig:quality}}
\end{figure*}

Despite the lack of supermodularity or submodularity, the accelerated variants of all the algorithms always perform akin to their non-accelerated variants on prediction \emph{and} profile error (Fig. \ref{fig:quality}).
This seems to indicate that the objective function may be close to supermodular in practice.
%Exploring how close the objective function is to super-modularity is left to future work.

\eat{
\begin{figure}
\centering
\subfigure[Prediction Error]{
   \includegraphics[height=3cm, width=3.8cm]{Figures/Final/EPS/netflix_5000users_3iter_20dim_100items.eps}
    %\caption{Average Group Satisfaction on top-$k$ Itemset Varying Number of Users}
     \label{fig:nf_algos8}
    }
\subfigure[Profile Error]{
   \includegraphics[height=3cm, width=3.8cm]
   {Figures/Final/EPS/netflix_PuError_5000users_3iter_20dim_100items.eps}
   %\caption{Average Group Satisfaction on top-$k$ Itemset Varying Number of Items}
     \label{fig:nf_pu}
    } 
\caption{\small Results on the Netflix Dataset \label{fig:netflix}}
\end{figure}

\begin{figure}
\centering
\subfigure[Prediction Error]{
   \includegraphics[height=3cm, width=3.8cm]{Figures/Final/EPS/ml20m_5000users_3iter_20dim_100items.eps}
  \label{fig:ml20m_algos8}
    }
\subfigure[Profile Error]{
   \includegraphics[height=3cm, width=3.8cm]
  {Figures/Final/EPS/ml20m_PuError_5000users_3iter_20dim_100items.eps}
  \label{fig:ml20m_pu}
    } 
\caption{\small Results on the Movielens 20M Dataset \label{fig:ml20m}}
\end{figure}
}

\subsection{Scalability Experiments}

To test scalability of our proposed solutions we run all $13$ algorithms on $2$ of the datasets, ML 100K and ML 1M under \emph{ideal} setting and on Netflix and ML 20M under \emph{real} setting, and measure running times with varying budget. Note that due to the datasets' sparsity, the average number of items per cold user that the algorithms sift through in the \emph{real} setting ranges from 263.8 to 281.4, while in the \emph{ideal} setting, it is significantly more (823 and 1833 for ML 100K and ML 1M respectively).

{\bf Results:} In all cases, the accelerated algorithms produce error similar to their un-accelerated counterparts (Fig. \ref{fig:quality}), but running time performance is far superior (Fig. \ref{fig:rt}). Among all algorithms, {\tt FG2} (both accelerated and unaccelerated) has the best qualitative performance,  with prediction and profile error comparable to {\tt BG2} (Fig. \ref{fig:quality}) or better, and is significantly faster than {\tt BG2} in terms of running time. In fact, even for ML 100K, our smallest dataset, under the \emph{ideal} setting, the time taken by unaccelerated {\tt FG2} for $b = 100$ is less than a sixth of the time taken by {\tt ABG2} for $b = 4$.
Moreover, running times of all backward greedy algorithms increase significantly as we decrease $b$ (see Fig. \ref{fig:rt}), which makes them unsuitable for use in a real world system, where $b$ would typically be very small.

\begin{figure*}[ht]
\centering
\subfigure[Running Time -- ML 100K]{
   \includegraphics[height=3cm, width=4cm] {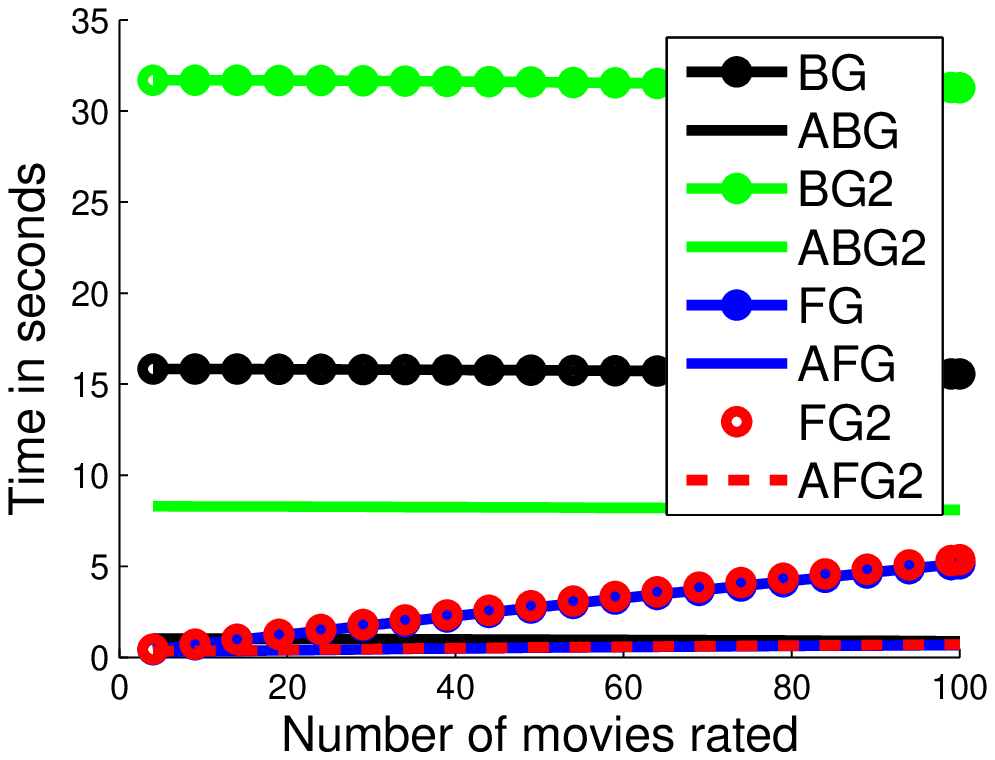}
    \label{fig:ml100k_comp}
 }
 \subfigure[Running Time -- ML 1M]{
   \includegraphics[height=3cm, width=4cm] {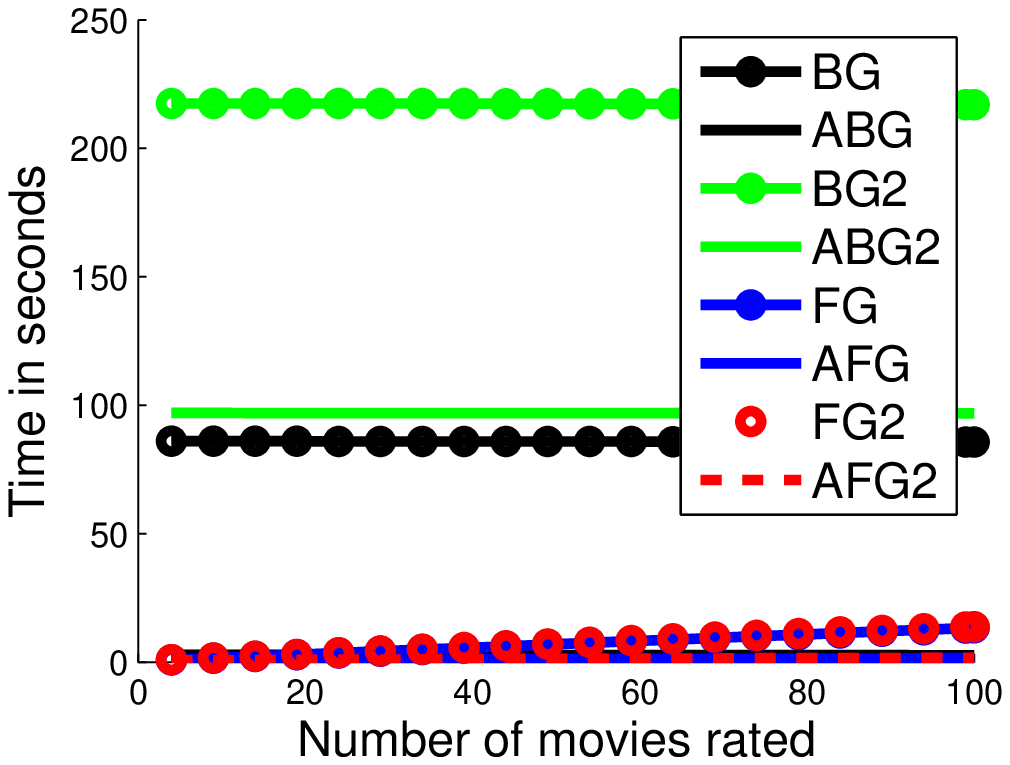}
    \label{fig:ml1m_comp}
 }
\subfigure[Running Time -- Netflix]{
   \includegraphics[height=3cm, width=4cm] {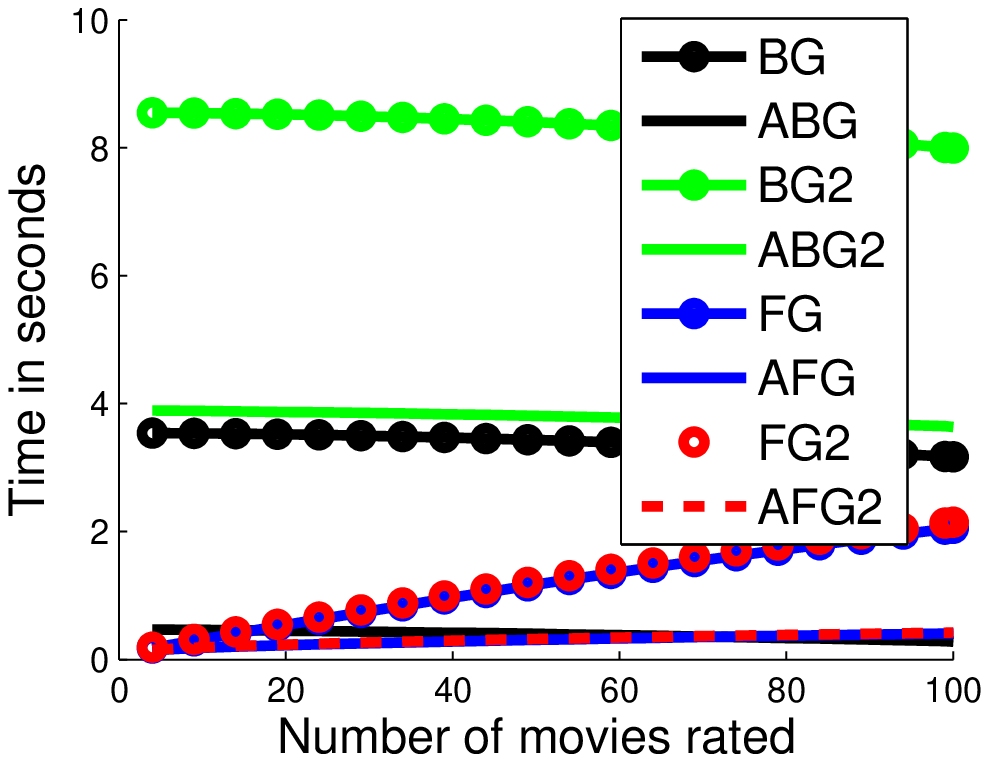}
    \label{fig:nf_comp}
 }
\subfigure[Running Time -- ML 20M]{
   \includegraphics[height=3cm, width=4cm] {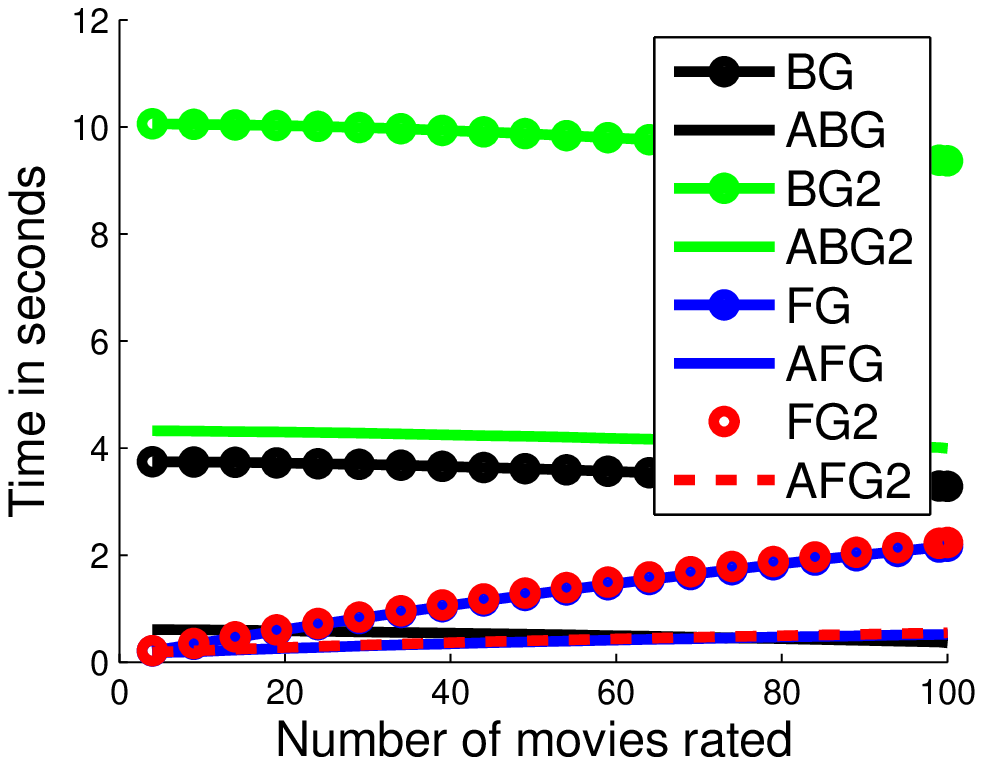}
    \label{fig:ml20m_comp}
 }
\caption{\small Running time versus budget size $b$ \label{fig:rt}}
\end{figure*}

%\subsection{Scalability Experiments}
\section{Conclusion}
\label{sec:conclusion}
In this paper, we consider model-based CF systems and investigate the optimal interview design problem for a cold-start user, that consists of a small number of items with which to interview and learn the user's interest. We formalize the problem as a discrete optimization problem to minimize the least square error between the true and estimated profile of the user, and present several non-trivial technical results. We present multiple non-trivial theoretical results including, NP-hardness, hardness of approximation, as well as proving that the objective function is neither submodular nor supermodular, suggesting efficient approximations are unlikely to exist. %Our counterexamples correct a previous misclaim in the literature \cite{anava2015budget} that the objective function is supermodular. 
To our best knowledge, a rigorous theoretical analysis of this problem has not been conducted before. We present several scalable heuristic algorithms and experimentally evaluate their quality and scalability performance on four large scale real datasets. Our experimental results demonstrate the effectiveness of our proposed (accelerated) algorithms and show that they significantly outperform previous algorithms while achieving a comparable profile error and prediction error performance. This is the first time a large scale experimental study involving large real datasets has been reported and it shows that unlike our proposed accelerated versions, previously proposed algorithms do not scale. As ongoing work, we focus on how to design a single interview plan for a batch of cold users.

%As future work, one can explore several interesting directions: instead of designing the interview plan for a single cold user, how to adapt our problem to design a single interview plan for a batch of cold users? One can also explore the optimization problem of  a ``hybrid'' interview plan that comprises multiple (small) batches of questions, where the next batch of questions is chosen after taking into account the cold-start user's feedback from previous batches. Finally, it would be interesting to analyze the complexity of bandit frameworks for interactive interviews of cold-start users, trading between exploration and exploitation, and develop scalable algorithms in such a framework. 

\bibliographystyle{ACM-Reference-Format}
\bibliography{short-names} 

%\appendix

%\input{appendix}

\end{document}